\documentclass[10pt, conference, letterpaper]{IEEEtran}
\IEEEoverridecommandlockouts
\usepackage{graphicx}
\usepackage{amssymb,amsfonts}
\usepackage{textcomp}
\usepackage{xcolor}
\usepackage{fancyhdr}
\usepackage{soul}
\usepackage{multirow}
\usepackage{diagbox} % 加载宏包
\usepackage{booktabs} % For formal tables
  
\usepackage[noend]{algpseudocode}  
\usepackage{algorithm}

\usepackage{url}
\usepackage{cite}
\usepackage{makecell}
\usepackage{threeparttable}
\usepackage{amsmath}
\usepackage{amsthm}
\usepackage{subfigure}
\usepackage{color}

\newtheorem{theorem}{Theorem}

\newtheorem{lemma}{Lemma}
\newtheorem{remark}{Remark}

\newcommand{\e}[1]{{\mathbb E}\left[ #1 \right]}
\newcommand{\etal}{\textit{et al}. }

\usepackage{enumitem}
\setenumerate[1]{itemsep=0pt,partopsep=0pt,parsep=\parskip,topsep=5pt}
\setitemize[1]{itemsep=0pt,partopsep=0pt,parsep=\parskip,topsep=5pt}
\setdescription{itemsep=0pt,partopsep=0pt,parsep=\parskip,topsep=5pt}

\def\BibTeX{{\rm B\kern-.05em{\sc i\kern-.025em b}\kern-.08em
    T\kern-.1667em\lower.7ex\hbox{E}\kern-.125emX}}

\begin{document}

% paper title
% can use linebreaks \\ within to get better formatting as desired
\title{On the Performance of Pipelined HotStuff}
\author{
\IEEEauthorblockN{Jianyu Niu, Fangyu Gai, Mohammad M. Jalalzai, Chen Feng}
\IEEEauthorblockA{School of Engineering, University of British Columbia (Okanagan Campus)}
%\IEEEauthorblockA{\IEEEauthorrefmark{2}Dapper Lab}
\IEEEauthorblockA{{\{jianyu.niu, fangyu.gai, m.jalalzai, chen.feng\}@ubc.ca}}
}

\maketitle
%\cfoot{}
%\pagestyle{fancy}
%\rfoot{\thepage}

\begin{abstract}
HotStuff is a state-of-the-art Byzantine fault-tolerant consensus protocol. It can be pipelined to build large-scale blockchains. 
One of its variants called LibraBFT is adopted in Facebook's Libra blockchain. Although it is well known that pipelined HotStuff is secure against up to $1/3$ of Byzantine nodes, its performance in terms of throughput and delay is still under-explored.
In this paper, we develop a multi-metric evaluation framework to quantitatively analyze pipelined \mbox{HotStuff's performance} with respect to its chain growth rate, chain quality, and latency. 
We then propose two attack strategies and evaluate their effects on the performance of pipelined HotStuff.
Our analysis shows that the chain growth rate (resp, chain quality) of pipelined HotStuff under our attacks can drop to as low as $4/9$ (resp, $12/17$) of that without attacks when $1/3$ nodes are Byzantine.
As another application, we use our framework to evaluate certain engineering optimizations adopted by LibraBFT. We find that these optimizations make the system more vulnerable to our attacks than the original pipelined HotStuff. 
Finally, we provide two countermeasures to thwart these attacks.
We hope that our studies can shed light on the rigorous understanding of the state-of-the-art pipelined HotStuff protocol as well as its variants.
%\color{red}{The three-chain definition to be solved later.}
\end{abstract}

%\begin{IEEEkeywords}
%Blockchain, Byzantine Fault Tolerant consensus, Pipelined HotStuff, LibraBFT.
%\end{IEEEkeywords}

%To do list
%---Complete the longest chain analysis
%---Analyze the optimal strategies 
%---Extend the scale of experiments, and showing the tradeoff 
\section{Introduction} \label{sec:intro}
In 2008, Nakamoto invented the concept of {blockchain}, a mechanism to maintain a distributed ledger for the cryptocurrency Bitcoin \cite{nakamoto2012bitcoin}.
The core novelty behind blockchain is Nakamoto Consensus (NC), an unconventional (at that time) synchronous Byzantine fault-tolerant (BFT) consensus \cite{garay2015,Pass2017}. 
Despite the huge impact of Bitcoin, NC suffers from long confirmation latency and low transaction throughput, both of which hinder the original blockchain to support Internet-scale applications. 
For example, Bitcoin today can only process up to seven transactions per second with a confirmation latency of hours.
On one hand, the long latency is a result of the probabilistic safety guarantee and no finality: a short latency cannot guarantee high confidence that a transaction has been confirmed.
On the other hand,  the low throughout is mainly due to the speed-security tradeoff: a higher transaction throughput leads to more severe forking, which greatly reduces the honest computation power against adversaries, making the system less secure~\cite{ghost}.

One promising approach to addressing these dilemmas is leveraging the classical BFT consensus~\cite{lamport1980,pbft1999}, which is also referred to as BFT state machine replication (SMR) and has been extensively studied for the last few decades. 
Unlike NC, classical BFT protocols can provide a strong safety guarantee. 
That is, once a transaction is confirmed, it will stay there forever.
Hence, clients do not need long waiting periods to confirm transactions (which means a shorter transaction latency), and transaction processing does not need to be compromised with security (which implies a higher throughput).
For example, experiments have demonstrated that PBFT~\cite{pbft1999}, a pioneer BFT protocol, can proceed tens of thousands of transactions per second in a LAN~\cite{Bessani2014} and have only hundreds of milliseconds latency in a WAN~\cite{Sousa2015}.
Despite all of these advantages, it is technically challenging to apply classical BFT protocols in a blockchain setup.
First, the classical BFT protocols have a high message complexity (e.g., $O(n^2)$ message complexity for committing one block), so the number of participants is usually less than dozens~\cite{QU,700BFT}. In other words, classical BFT protocols suffer from scalability issues and cannot support a large-scale blockchain. 
%For example, deployments and experiments of BFT protocols often employ a minimum of four participants~\cite{kotla2007zyzzyva}, and generally have not explored scalability levels beyond a few tens of participants~\cite{scalaBFT, HQReplica,700BFT}.
Second, classical BFT protocols are notoriously difficult to be developed, tested, and proved~\cite{mickens2014saddest, 700BFT, QuestMarko}.
%A design flaw in a famous BFT protocol, Zyzzyva~\cite{kotla2007zyzzyva}, has surfaced a decade after its original proposal~\cite{abraham2018revisiting}.
Finally, classical BFT protocols rarely consider fairness among leaders\footnote{
%The leader fairness is the foundation of fairly incentivizing nodes to 
Blockchain systems usually reward leaders with some self-issued tokens for incentivizing protocol participation~\cite{nakamoto2012bitcoin}. Hence, leadership fairness is the foundation of such an incentive mechanism to fairly reward nodes.}. 
In most leader-based BFT protocols~\cite{pbft1999}, a node can serve as a leader as long as it behaves well. This is also called stability-favoring leader rotation~\cite{chan2018pala}, for this mechanism can avoid the $O(n^3)$ message complexity in the leader rotation.

%To address these challenges, various state-of-the-art BFT protocols are proposed including Tendermint~\cite{buchman2016tendermint}, Casper FFG~\cite{casper}, Pala~\cite{chan2018pala}, HotStuff~\cite{HotStuffYin2019}.
%Among them, HotStuff proposed by Yin et al.~\cite{HotStuffYin2019} is the first to solve all these issues.
HotStuff proposed by Yin et al.~\cite{HotStuffYin2019} is a state-of-the-art BFT consensus, {which leverages the community's advances in the last several decades and achieves the strong scalability, prominent simplicity, and good practicability for large-scale applications like blockchains.}
HotStuff creatively adopts a three-phase commit rule (rather than the two-phase commit rule used in classical BFT~\cite{pbft1999}) to enable the protocol to reach consensus at the pace of actual network delay\footnote{This property is called responsiveness in the literature.} and leverages the threshold signature to realize linear message complexity.
HotStuff can be further pipelined, which enables a frequent leader rotation and leads to a simple and practical approach to building large-scale blockchains.
Due to these salient properties, Facebook adopts a variant of pipelined HotStuff called LibraBFT~\cite{banostate} for its global payment system, Libra blockchain, {which aims to thrive fintech innovations and to enable billions of consumers and businesses to conduct instantaneous, low-cost, highly secure transactions.}\footnote{Pipelined HotStuff is also referred to as chained HotStuff~\cite{HotStuffYin2019}. } 
In addition, Dapper Lab describes how to deploy HotStuff in its Flow platform~\cite{Hentschel2020FlowSC}, and Cypherium Blockchain\cite{cypherium} combines HotStuff with NC together to build a permissionless blockchain. 
Although it is well known that pipelined HotStuff is secure against up to $1/3$ of Byzantine nodes, its performance in terms of throughput and delay is still under-explored.

In this paper, we first develop a multi-metric evaluation framework to quantitatively analyze pipelined HotStuff's performance with respect to its chain growth rate, chain quality, and latency. 
We then propose several attack strategies and evaluate their effects on the performance by using our framework.
In addition, we use our framework to evaluate some engineering optimizations adopted by LibraBFT. We find that these optimizations make the system more vulnerable to certain attacks compared with the original pipelined HotStuff. Finally, we provide two countermeasures to thwart these attacks. We hope that our studies can shed light on the rigorous understanding of the state-of-the-art pipelined HotStuff protocol as well as its variants. Our contributions can be summarized as follows:

\iffalse
In this paper, 
%we first present an attack, called forking attack, in which the adversary creates block forks on purpose to override blocks proposed by honest nodes.
we first systemically evaluate the performance of the pipelined HotStuff under attacks.
We provide a multi-metric evaluation framework including chain growth rate, chain quality, and latency to quantitatively analyze the impact of attacks.
We find that some engineering ``optimizations'' adopted by LibraBFT may make it more vulnerable to the attacks.
%although both of them are vulnerable to the new forking attack, pipelined HotStuff can achieve a better performance than the pipelined HotStuff. 
For example, in pipelined HotStuff, an adversary, controlling $1/3$ corrupted nodes can increase the latency to $9.93$ rounds on expectation ($3$x times of the latency without attacks), however, with the same condition, the adversary in LibraBFT can increase the latency to $10.68$ rounds. 
%For example, an adversary, controlling $1/3$ corrupted nodes can decrease the chain quality of pipelined HotStuff to $1/X$ of that without any attack; however, with the same condition, the adversary can only lower the chain quality to $1/Y$.
We also provide multiple countermeasures to thwart the attacks and evaluate their efficiencies.
We hope our studies can shed light on the rigorous understanding of the state-of-the-art pipelined HotStuff protocol.
\fi

\begin{itemize}[leftmargin=*]
    \item We develop a multi-metric evaluation framework and leverage it to evaluate the impact of several new attacks. Our analysis shows that the chain growth rate (resp, chain quality) of pipelined HotStuff under these attacks can drop to $4/9$ (resp, $12/17$) of that without attacks when $1/3$ nodes are Byzantine.
    
    \item We use our framework to evaluate some engineering optimizations adopted by LibraBFT. We find that in pipelined HotStuff, an adversary controlling $1/3$ corrupted nodes can increase the latency to $8.33$ rounds on expectation ($2.7$x times of the latency without attacks), however, with the same condition, the adversary in LibraBFT can increase the latency to $10.25$ rounds. 
    
    \item We propose two countermeasures against our attacks, which can reduce the latency by $3$ rounds, improve the chain growth rate by $1.5$x times, and chain quality by $1.2$x times.
    
    \item We develop a proof-of-concept implementation of pipelined HotStuff to validate our theoretical findings. 
    %We hope our studies can shed light on the rigorous understanding of the state-of-the-art pipelined HotStuff protocol.
\end{itemize}
\iffalse
\noindent The rest of the paper is organized as follows.
Section~\ref{sec:model} provides the system model and some preliminaries.
Section~\ref{sec:algorithm} gives explicit algorithms for pipelined HotStuff, LibraBFT, and the forking attacks. 
Section~\ref{sec:analysis} gives the mathematical analysis of several important performance metrics. 
Section~\ref{sec:countermeasure} presents the countermeasure for the attacks, and analyze their efficiencies.
The evaluation is provided in Sec.~\ref{sec:evaluation}, and the related work is discussed in Section~\ref{sec:related}.
Finally, Section~\ref{sec:conclusion} concludes the paper.
\fi

\section{System Model and Preliminaries} \label{sec:model}
%In this section, we first introduce the system model by following the prior works \cite{dwork1988consensus, castro1999practical, HotStuffYin2019},  and then present the leader-based round (LBR) abstraction and some preliminaries.

\subsection{System Model}
We consider a system with $n$ nodes denoted by the set $\mathcal{N}$. 
We assume a public-key infrastructure (PKI), and each node has a pair of keys for signing messages (e.g., blocks and votes).
We assume that a subset of $f$ nodes is \emph{Byzantine}, denoted by the set $\mathcal{F}$, and can behave arbitrarily. 
The other nodes in $\mathcal{N} \setminus \mathcal{F}$ are honest and strictly follow the protocol. 
In order to ensure security, we have $n \ge 3f + 1$.
We use $\alpha$ (resp. $\beta$) to denote the fraction of Byzantine (resp. honest) nodes. That is, $\alpha = f/n$.
For simplicity, all the Byzantine nodes are assumed to be controlled by a single adversary, which is computationally bounded and cannot (except with negligible probability) forge honest nodes' messages. 
%In other words, the used cryptographic tools are perfect. 

We assume honest nodes are fully and reliably connected, i.e., every pair of honest nodes is connected with an authenticated and reliable communication link.
%one honest node can always send messages to other honest nodes via point-to-point or broadcast pattern. 
We adopt the partial synchrony model of Dwork \etal \cite{dwork1988consensus}.
In the model, there is a known bound $\Delta$ and an unknown Global Stabilization Time (\textsf{GST}), such that after \textsf{GST}, all message transmissions between two honest nodes arrive within a bound $\Delta$. 
Hence, the system is running in \emph{synchronous} mode after \textsf{GST} and \emph{asynchronous} mode if \textsf{GST} never occurs. 

\subsection{Preliminaries} \label{subsec:preliminaries}
\noindent \textbf{Quorum Certificate.}  A block's quorum certificate (QC) is proof that more than $2n/3$ nodes (out of $n$) have signed this block. Here, a QC could be implemented as a simple set of individual signatures or a threshold signature.
We say a block is certified when its QC is received and certified blocks' freshness is ranked by their round numbers.
In particular, we refer to a certified block with the highest round number that a node knows as the \emph{newest} certified block. 
Each node keeps track of all signatures for all blocks and keeps updating the newest certified block to its knowledge.

\noindent \textbf{Block and Block Tree.} Clients send transactions to leaders, who then batch transactions into blocks. 
A block has four-tuple $\left \langle round, cmd, parent\_qc, \sigma \right \rangle$, where $round$ denotes the round number at which the block is proposed, $cmd$ is a batch of transactions, $parent\_qc$ is the QC for the parent block, and $\sigma$ is the block owner's signature of $\left \langle round, cmd, parent\_qc\right \rangle$. Every block except the genesis block must specify its parent block and include a QC for the parent block. In this way, blocks are chained. 
(Note that in Bitcoin~\cite{nakamoto2012bitcoin}, blocks are chained through hash references rather than QCs.)
%A block’s height is its distance from the genesis block.
%position in the chain is referred to as its height.
As there may be forks, each node maintains a block tree (referred to as $blockTree$) of received blocks. 
%However, there is only one chain containing the committed blocks, and this chain is referred to as the \emph{main} chain.
\begin{figure}[t]
\centering
\includegraphics[width=3.2in]{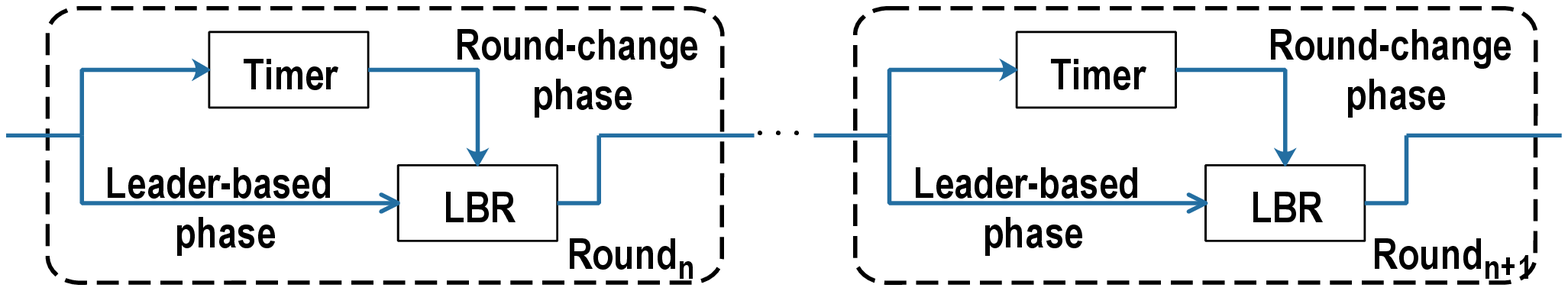}
\caption{\textbf{Overview of a sequence of the leader-based round (LBR) instances.} The leader-base phase is in charge of driving progress, while the round-change phase is to synchronize nodes to the same round.}
\vspace{-2mm}
\label{fig:view}
\end{figure}

\subsection{Leader-based Round Abstraction} 
Pipelined HotStuff is executed into a sequence of rounds, and each round has a designated leader\footnote{Rounds are also referred to as views~\cite{pbft1999, HotStuffYin2019}, terms~\cite{raft2014}, instance values/epoch~\cite{ZooKeepper} or ballot numbers in the literature}. 
Each round can be further divided into two phases: $1$) leader-based phase in which a designated leader proposes a new block and collects votes from other nodes to form a quorum certificate of this block (introduced shortly), and $2$) round-change phase in which nodes safely \emph{wedge} to next round if the current-round leader is faulty or no certified block is generated before the timeout, as shown in Fig.~\ref{fig:view}.
By following recent work~\cite{spiegelman2019ace}, we assume that each round can be encapsulated in a leader-based round (LBR) abstraction, which provides two important modules: pacemaker and leader-election modules.
%First, the round-change phase and synchronization logic can be abstracted into a module named \emph{pacemaker} by following pipelined HotStuff~\cite{HotStuffYin2019}.
The pacemaker module can guarantee that honest validators are synchronized to execute the same rounds for sufficient overlap, and leaders propose a block that will be supported by honest nodes when the network is synchronous~\cite{HotStuffYin2019}.
That is, an honest leader can send its block proposal to all the other honest nodes and receive their votes in one round.
The leader-election module can guarantee that nodes are \emph{fairly} elected as leaders.
That is, during synchronous periods, each node has the same chance to win the leadership for one round\footnote{During asynchronous periods, some honest nodes may not be synchronized to the highest round and participate in the leader-election. Thus, the adversary can have a higher chance than $\alpha$ to be elected as leaders.}. For convenience, a leader who is elected from honest nodes (resp. Byzantine nodes) is referred to as \emph{honest} (resp. \emph{adversarial}) leader. In the same way, a block proposed by an honest (resp. adversarial) leader is referred to as \emph{honest} (resp. \emph{adversarial}) block.

\section{Pipelined HotStuff and LibraBFT Algorithms} \label{sec:algorithm}
%In this section, we first describe the algorithms of pipelined HotStuff, also called chained HotStuff~\cite{HotStuffYin2019}, and the variant LibraBFT. 
\subsection{Pipelined HotStuff Algorithm} \label{subsec:hotstuffalgorithm}
%To better understand them, we first formalize the procedure for one round.
We describe the leader-based phase of pipelined HotStuff at round $i$. 
In the beginning, a unique leader, called $leader_i$, is randomly elected by the leader-election module. The leader's identity is known and can be verified by all nodes (see Sec.~\ref{subsec:preliminaries}).
%For clarity, we assume the leader election is abstracted and provided some well-defined functions.
%In addition, the leader can update its states from at least $2/3n$ nodes.
Then, the leader {\bf proposes} a block to extend the \emph{newest} certified block it has seen\footnote{For simplicity, we follow the same way with LibraBFT, i.e., leaders extend the predecessor block with a direct child. However, in pipelined HotStuff, a leader has to include dummy blocks in its proposal if there are no certified blocks generated in previous rounds.}, and broadcasts this block to all other nodes. Every node {\bf votes} for the first block it receives from the leader\footnote{Recall that an adversarial leader can propose multiple blocks.}, as long as the block satisfies certain conditions introduced shortly. A vote for a block is a signature on the block.
When the leader receives at least $2n/3$ unique signatures (including its own signature), it aggregates these signatures into a QC and sends it to the next-round leader, namely, $leader_{i+1}$. 

\begin{figure}[t]
\centering
\includegraphics[width=2.5in]{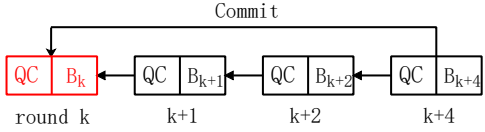}
\caption{\textbf{The chained blocks in pipelined HotStuff.} Blocks $B_k$, $B_{k+1}$, and $B_{k+2}$ are three consecutive blocks, and nodes will commit block $B_k$ when receiving block $B_{k+4}$.}
\vspace{-2mm}
\label{fig:commit}
\end{figure}

We are now ready to describe the condition for voting. Every node maintains two local parameters: $i$) $last\_voted\_round$, the last round for which the node has voted, and $ii$) $locked\_round$, the highest known round number of the grand-parent block that has been received.
For example, in Fig.~\ref{fig:commit}, when a node first receives the block $B_{k+2}$ and votes for it, its $last\_voted\_round$ is updated to $k+2$, and the newest grandparent block is $B_k$, and so $locked\_round$ is $k$. 
After receiving block $B_{k+4}$, the node updates $last\_voted\_round$ to $k+4$ and $locked\_round$ to $k+1$.
A node will vote for the first block proposed by the current-round leader if the block extends the block generated at $locked\_round$ (regardless of round numbers) or the round number of the received block's parent block is greater than $locked\_round$. 
Concretely, a block satisfies the above condition if and only if: $i$) its round number is greater than $last\_voted\_round$, and $ii$) the round number of its parent block is greater than or equal to $locked\_round$.\footnote{This condition is adopted by LibraBFT and is shown to be equivalent to the previous condition~\cite{banostate, bano2020twins}.}

After voting for a new block, a node will insert the block to its $blockTree$ and then update its state as follows: $i$)update $last\_voted\_round$ to the round number of this block, and $ii$) update the node’s $locked\_round$ to the round number of this block' grandparent block if the latter is higher.
Meanwhile, the node checks whether there are new committed blocks in its $blockTree$. Specifically,
if there are three blocks $B_k$, $B_{k+1}$ and $B_{k+2}$ proposed in three consecutive rounds $k$, $k+1$, and $k+2$, and an additional block extends block $B_{k+2}$, the node will {\bf commit} block $B_k$ and all its predecessor blocks. The first three consecutive blocks are referred to as $3$-direct chain\footnote{A $3$-direct chain requires an additional block extending $3$ consecutive blocks. If we only have $3$ consecutive blocks, we don't call them a $3$-direct chain. Note that for simplicity of analysis, the $3$-direct chain is different from the Three-Chain defined in ~\cite{HotStuffYin2019}, but the differences do not affect the results.}.
A simple case, in which block $B_k$ is committed, is shown in Fig.~\ref{fig:commit}. 
Note that nodes maintain a chain containing committed blocks, and this chain is referred to as the \emph{main} chain in our later analysis.
%The block $B_{k+1}$ cannot be committed for not satisfying the committing rule.

To sum up, in pipelined HotStuff, leaders propose new blocks, and every node votes for a new block according to certain voting conditions and sometimes commits blocks if there is a $3$-direct chain followed by another block. Also, every node starts a timer to track progress for each round.
Whenever timeout happens or a block's QC is received, a node moves to the next round. 
Such a round synchronization procedure is provided by the pacemaker module.
In fact, the pacemaker module can also guarantee a new leader to have the newest certified block and/or its parent block, which guarantees the leader can propose a block voted by honest nodes (see Sec.~\ref{subsec:preliminaries}).
%When the timer expires and a node still has not received a proposal, it broadcasts a timeout vote on a Nil block. When a node gathers enough timeout votes to form a timeout certificate, it advances its round. Every time a round fails, timeout periods are increased, allowing lagging nodes to catch up and enabling the protocol to eventually reach a decision.

\subsection{LibraBFT Algorithm} \label{subsec:librabft}
LibraBFT is a variant of pipelined HotStuff with two subtle differences.
\emph{First}, in LibraBFT, a node sends its vote directly to the next-round leader (rather than the current-round leader) so that the next leader can form a QC and embed the QC in its own block. In this way, the current leader doesn't need to relay the QC to the next leader. That is, 
this optimization can cutoff the delay of the relay operation.\footnote{The latest version of HotStuff also adopted this optimization~\cite{hotstuffLatest}.}
\emph{Second}, LibraBFT introduces a new block type called Nil block. 
This is, when the timer expires, and nodes have not received a proposal for the round, they can broadcast a vote on a Nil block (in a predetermined format).
If more than $2n/3$ nodes have voted, the aggregated signatures can serve as a QC for the Nil block.
The certified Nil block can guarantee that blocks are produced in consecutive rounds despite having faulty leaders, which can further accelerate the block commitment.
For example, in Fig.~\ref{fig:commit}, when the leader of round $k+3$ is faulty, and there is no block produced, nodes cannot commit block $B_{k+1}$ even if receiving block $B_{k+4}$ in pipelined HotStuff. By contrast, in LibraBFT, there will be a certified Nil block at round $k+3$, and nodes will commit block $B_{k+1}$ after receiving block $B_{k+4}$.

\section{Performance Metrics and Attacks}
In this section, we first introduce a multi-metric framework to evaluate the impact of various attacks and then propose several attack strategies. 
\subsection{Performance Metrics}
We focus on three performance metrics, namely, chain growth rate, chain quality, and latency. All of these metrics are meaningful only after
the \textsf{GST}, which implies that the network is in synchronous mode\footnote{Before the \textsf{GST}, there may have no certified blocks at all, and so the three metrics become meaningless.}.
%which implies network is synchronous\footnote{By the well-known FLP impossibility~\cite{flp}, pipelined HotStuff is designed to satisfy safety in an asynchronous network, however, consider a “long enough” periods of network synchrony to make progress.}.
%Fisher, Lynch and Paterson (FLP)~\cite{flp} proved that a deterministic consensus protocol in an asynchronous network cannot guarantee liveness even if one node can fail; 

\subsubsection{Chain Growth Rate}
For a given adversarial strategy that controls a fraction $\alpha$ of total nodes, the chain growth rate $u_{1}(\alpha)$ is defined as the rate of honest blocks appended to the main chain over the long run. Let $B_h(m)$ denote the total number of honest blocks appended to the main chain in $m$ rounds. (Note that $B_h(m)$ can be a random variable because of the randomness in the leader selection.) We have:
\begin{equation}
    u_{1}(\alpha) = {\lim_{m \rightarrow \infty}{\frac{B_h(m)}{m}}}.
\end{equation}
The chain growth rate corresponds to the liveness in the context of blockchains.
%The chain growth rate is measured by only counting the honest blocks, and it corresponds with the liveness in the context of blockchain. In addition, $u_1(\alpha)$ depends on the specific attack.
%In the prior works of NC~\cite{garay2015,Pass2017}, the liveness is usually decomposed into chain growth (i.e., including both the honest and adversarial blocks) and chain quality (i.e., the fraction of the honest block among all blocks). 
%With chain growth and chain quality, it is easy to derive chain growth rate; however, the chain growth rate is more convenient.
%Here, the chain growth rate has a more direct relationship with liveness. In addition, once some system parameters are set (e.g., the block size, %transaction size, and the average round interval), transaction throughput can be further computed from the chain growth rate. 

\subsubsection{Chain Quality} 
For a given adversarial strategy that controls a fraction $\alpha$ of total nodes, the chain quality $u_{2}(\alpha)$ is defined as the fraction of honest blocks included in the main chain over the long run. Let $B_a(m)$ denote the total number of adversarial blocks appended to the main chain in $m$ rounds. We have:
\begin{equation}
    u_{2}(\alpha) = {\lim_{m \rightarrow \infty}{\frac{B_h(m)}{B_h(m) + B_a(m)}}}.
\end{equation}
This metric affects the reward distribution.
%is usually related to the reward distribution policy. 
In blockchains, each block in the main chain brings its proposer a reward~\cite{nakamoto2012bitcoin}.
This reward incentivizes nodes to participate in the consensus and compete to win the leadership.
Additionally, nodes are expected to get rewards, proportional to their devoted resources (e.g., hash power in Poof-of-Work~\cite{nakamoto2012bitcoin}, and stakes in Proof-of-Stake~\cite{Gilad2017}). 
%For example, the block owner in Bitcoin can get a block reward, whereas in
Intuitively, a chain quality less than $1-\alpha$ implies that the adversary can win a higher fraction of rewards than what it deserves, which ruins the incentive compatibility~\cite{Pass2017fruit, eyal2014majority, niu2019selfish, niu2020incentive}. 

\subsubsection{Latency} For a given adversarial strategy that controls a fraction $\alpha$ of total nodes, the latency $u_{3}(\alpha)$ is defined as the average rounds that honest blocks take from being included in the main chain until being committed over the long run. Let $D_i$ denote the number of rounds that the $i$-th honest block takes to be committed by all honest nodes (rather than some honest leaders) during the $m$ rounds.\footnote{Note that leaders first commit blocks locally, and then send out the proofs of the $3$-direct chain to convince other nodes to commit these blocks. 
It is trivial to extend our analysis to get the results under this model.} 
We have:
\iffalse
\begin{equation}
    u_{3}(\alpha) = {\lim_{m \rightarrow \infty}{\frac{\sum_{i=1}^{n}{D_i}}{n}}}.
\end{equation}
\fi
\begin{equation}
    u_{3}(\alpha) = {\lim_{m \rightarrow \infty}{\frac{\sum_{i=1}^{B_h(m)}{D_i}}{B_h(m)}}}.
\end{equation}

\begin{remark}
Note that the chain growth rate and latency are measured in terms of \emph{rounds} rather than in time. This round abstraction allows us to ignore the specific implementation of the pacemaker module and focus on the core of pipelined HotStuff.
%This is because performance measurements in the time heavily rely on the adopted parameter settings (e.g., propagation delay, block size, and network bandwidth) and the pacemaker realization.
%By contrast, round abstraction and measurements in round enable us to focus on the core of pipelined HotStuff.
%can provide convenience and fair comparisons between different BFT protocols.
%In the same way, the chain growth rate is also measured in rounds. 
\end{remark}

\subsection{Attack Strategies} \label{subsec:attacks}
We introduce two attacks here. The forking attack launched by the adversary aims to minimize the chain growth rate and the chain quality by overriding honest blocks. The delay attack aims to maximize the latency by delaying the commitment of honest blocks. Both attacks, which are inspired by the selfish mining attack for Bitcoin, are new in the context of pipelined HotStuff. The optimality of these attacks will be discussed in a journal version of this work. Here, we emphasize that since the performance metrics are measured after the \textsf{GST}, we do not need to consider network-level attacks, such as eclipse attack\cite{eclipse2015} and Distributed Denial of Service (DDoS) attack, which may cause a network partition (i.e., asynchronous network) and blocking the progress.

%first describe an optimal attack called forking attack by which the adversary can override honest blocks to lower the chain growth rate and the chain quality\footnote{
%Under the synchronous network condition, we do not consider network attacks, e.g., eclipse attack\cite{eclipse2015} and Distributed Denial of Service (DDoS) attack, which may cause a network partition (i.e., asynchronous network), blocking the progress, and are largely orthogonal to our attack strategies.}.
%We then introduce a delay attack, which aims to maximize block latency.
%By contrast, delaying the honest blocks commitment, the strategies needs both timeout and overriding. 
\begin{figure}[t]
\centering
\includegraphics[width=2.2in]{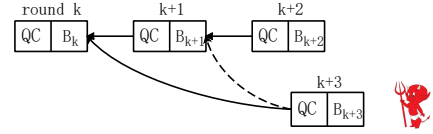}
\caption{\textbf{The forking attack on pipelined HotStuff.} The adversary is elected as a leader in round $k+3$. It proposes a block $B_{k+3}$ after block $B_{k}$ (or $B_{k+1}$) to override blocks $B_{k+1}$ and $B_{k+2}$ (or block $B_{k+2}$).}
\vspace{-2mm}
\label{fig:overriding}
\end{figure}
\subsubsection{Forking Attack}
In pipelined HotStuff, an adversarial leader can create forking blocks on purpose to override honest blocks without any loss. 
For example, in Fig.~\ref{fig:overriding}, if blocks $B_{k}$ and $B_{k+1}$ are both honest blocks, and the adversarial leader at round $k+3$ has no adversarial certified block with a round number larger than $locked\_round = k$, the adversarial leader will build a block on block $B_k$.
As block $B_{k+3}$ satisfies the voting condition (i.e., the round number of $B_{k+3}$'s parent block is no less than honest nodes' $locked\_round = k$), nodes will vote for $B_{k+3}$ and all subsequent leaders will extend $B_{k+3}$.
Similarly, if only block $B_{k+2}$ is an honest block, the adversary can build on block $B_{k+1}$  to override this block.
In both cases, the adversarial leader overrides some honest blocks and suffers no loss of adversarial blocks. Also, note that the adversarial leader cannot override block $B_k$ and its predecessor blocks, since block $B_{k+3}$ cannot reference a certified parent block, which has a round number no less than $k$.
%Note that in Fig.~\ref{fig:overriding}, 
%In addition, when the adversary is elected as a leader and proposes a block, it can push the duration of one round close to the timeout. In this way, it can slow the block growth rate. 
%In addition, it is easy to see that the adversary overrides honest blocks $B_{k+2}$ and/or $B_{k+1}$ without any loss. 
In general, if there exist some adversarial certified blocks with round numbers no less than $locked\_round$, the adversarial leader extends the newest adversarial certified block. Otherwise, the adversarial leader extends the block produced at $locked\_round$. 

\subsubsection{Delay Attack} \label{subsec:delayattack}
The main goal of the delay attack is to break the block commitment condition in order to increase the average delay of honest blocks in the main chain. 
More specifically, the delay attack is to prevent honest blocks to form the $3$-direct chain.
Recall that a block is committed if and only if three blocks extend it, and the first two blocks are produced in consecutive rounds after the block's round (i.e., the $3$-direct chain structure). 
Note that once the block is committed, all its predecessor blocks are also committed. 

\emph{Delay Attack in pipelined HotStuff.} Recall that an honest leader always proposes a block on the newest certified block. By contrast,
an adversarial leader can propose a block on a \emph{non-newest} certified block or propose \emph{no} block at all.
An example is provided in Case A of Fig.~\ref{fig:delayAttack} in which the adversarial leader of round $k+4$ observes three non-consecutive blocks
$B_{k}$, $B_{k+2}$, and $B_{k+3}$. In this case, the leader proposes no block at all (leading to a timeout). As a result, subsequent honest leaders have to restart building a $3$-direct chain.
Another example is illustrated in Case B of Fig.~\ref{fig:delayAttack}
in which the adversarial leader of round $k+4$ observes three consecutive blocks $B_{k+1}$, $B_{k+2}$, and $B_{k+3}$. In this case, the leader
proposes block $B_{k+4}$ on top of $B_{k+2}$. This will override block $B_{k+3}$ as explained before.
In general, if there exist three consecutive blocks that ended with the newest certified block, the subsequent adversarial leader overrides the newest certified block. Otherwise, the adversarial leader proposes no block.
%(To Jianyu: Please complete this part. You can follow the style of the forking attack: introduce a general rule after describing some examples.)
%Note that in this case, if the adversarial leader still pretends to be faulty, the first block $B_k$ in the $3$-direct chain will be committed once there is one new honest block appended to block $B_{k+2}$.  
%(See the case of committed block $B_k$ in Fig.~\ref{fig:commit}.)
\begin{figure}[t]
\centering
\includegraphics[width=2.2in]{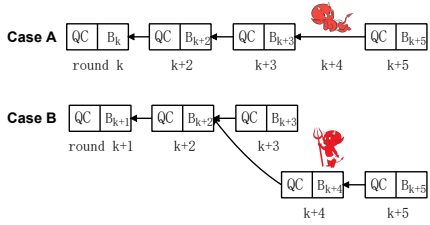}
\caption{\textbf{The delay attack on pipelined HotStuff.} The adversary proposes no block at all (in Case A) or overrides honest blocks at round $k+4$ (in Case B) according to whether there exists three consecutive blocks.}
\vspace{-2mm}
\label{fig:delayAttack}
\end{figure}

\iffalse
\begin{algorithm}[t!]
\caption{The delay attack strategy in pipelined HotStuff}\label{alg:delayattack}
\begin{algorithmic}[1]
%\Statex \textbf{Init:} Select one latest honest block as target block; 
\Statex \textbf{on} Some honest miners elected as the leader:
\State\hspace{\algorithmicindent} Propose a new certified block
\State\hspace{\algorithmicindent} \textbf{if} target block is committed \textbf{then}
\State\hspace{\algorithmicindent}\hspace{\algorithmicindent} 
Update the target block
\State\hspace{\algorithmicindent} \textbf{end}

\Statex \textbf{on} The adversary is elected as the leader:
\State\hspace{\algorithmicindent} \textbf{if} there is no $3$-direct chain \textbf{then}
\State\hspace{\algorithmicindent}\hspace{\algorithmicindent} wait for timeout
\State\hspace{\algorithmicindent} \textbf{else} 
\State\hspace{\algorithmicindent}\hspace{\algorithmicindent} override the last two blocks
\State\hspace{\algorithmicindent} \textbf{end}
\end{algorithmic}
\end{algorithm}
\fi 

\emph{Delay Attack in LibraBFT.}  
Delay attack in LibraBFT is slightly different from that in pipelined HotStuff.
First, due to the Nil block, the adversary cannot propose no block, because otherwise a certified Nil block will be produced that can be part of a $3$-direct chain. 
Hence, the adversarial leader can create a block, and then send this block to half of the honest nodes. 
In this way, half of the honest nodes will vote for this block, while the left half honest nodes will vote for the Nil block.
As a result, neither a certified block nor a certified Nil block will be produced in this round.
Second, in LibraBFT, as nodes send block votes to the next leader, an adversarial leader can hide the collected QC of a block proposed in the previous round. So, when there are three consecutive blocks, the adversarial leader can hide the QC for the last one. In this way, subsequent honest leaders cannot form a $3$-direct chain based on the three consecutive blocks.
%Therefore, when there is a $3$-direct chain, the adversarial leader can just hide the collected QC, and meanwhile make no certified block produced.
%As a result, subsequent honest leaders cannot extend the $3$-direct chain without a QC, and so have to rebuild the $3$-direct chain. 

\section{Performance Analysis under Forking and Delay Attacks} \label{sec:analysis}
In this section, we first analyze the performance of pipelined HotStuff in terms of chain growth rate, chain quality, and latency under forking and delay attacks. 
Then, we evaluate some optimizations adopted by LibraBFT.  
Specifically, we consider a sequence of $m$ rounds (when the network is in synchronous mode) and number these rounds as $1, 2, ..., m$. 
For each round, the possibility that the elected leader is honest (resp, adversarial) is $\beta$ (resp, $\alpha$). 
Now, Let $X_j$ ($j \in [1, m]$) denote an indicator random variable which equals one if the leader of the $j$th round is honest and equals zero otherwise. 
As the network is synchronous, nodes can receive a block within $\Delta$ time after an honest leader sends the block. 
For simplicity, nodes are assumed to receive the block by the end of each round. 
In addition, honest leaders are assumed to be able to get the newest honest certified blocks from other nodes before proposing new blocks.

\subsection{Performance Analysis of Pipelined HotStuff }
\subsubsection{Chain Growth Rate}
%We compute the chain growth rate of pipelined HotStuff. In particular,
Recall that an honest block proposed at round $i$ will be overridden by an adversarial leader in round $(i+1)$ or $(i+2)$ under the forking attack (See Sec.~\ref{subsec:attacks} for details.). 
In other words, an honest block can be kept in the main chain if and only if the subsequent two blocks are honest blocks.
Let $Z_i$ denote an indicator random variable, which equals to one if $\{ X_i = 1, X_{i+1} = 1, X_{i+2} = 1\}$ and equals to zero otherwise. 
Next, let $Z = \sum_{i=1}^{m}Z_i$. The following lemma bounds the value of $Z$.

\begin{lemma}\label{lem:keylemma1}
For $m$ consecutive rounds, the number of block fragments $(X_i, X_{i+1}, X_{i+2}) = (1,1,1)$ has the following Chernoff-type bound: For $0 < \delta < 1$,
\begin{equation}
   \Pr( | Z - \beta^3 m| > \delta \beta^3 m ) < e^{-\Omega\left(\delta^2 \beta^3  m\right)}. 
\end{equation}
\end{lemma}

\begin{proof}
Without loss of generality, we assume that $m$ is a multiple of $3$. 
Let $Z^{j} = \sum_{i=0}^{m/3-1} Z_{j + 3i}$ ($j \in [0, 1, 2]$).
Then, $Z = Z^{0} + Z^{1} + Z^{2}$. 
It is easy to show that $E\left( Z^{j} \right) = \beta^3 m/3$, since $P\{Z_i = 1\} = 
P\{ X_i = 1\} P\{ X_{i+1} = 1 \} P\{ X_{i+2} = 1 \}= \beta^3$.
Note that $\{ Z_{0}, Z_{3}, \ldots, Z_{m-1} \}$ are independent random variables, because $Z_{i}$ is a function of
$(X_{i},X_{i+1},X_{i+2})$. Hence, $Z^{j}$ is a sum of i.i.d. random variables. 
By Lemma~$4$ in~\cite{HotStuffGit}, we have 
\[\Pr \left( Z < (1 - \delta) \beta^3  m \right) <  e^{-\delta^2 \beta^3  m /6} = e^{-\Omega\left(\delta^2 \beta^3  m \right)}.
\]
Similarly, we have $\Pr \left( Z > (1 + \delta) \beta^3 m\right) < e^{-\delta^2 \beta^3 m / 9} = e^{-\Omega\left(\delta^2 \beta^3 m\right)}$. 
\end{proof} 

This lemma shows that as $m$ increases, the number of block fragments $(X_i, X_{i+1}, X_{i+2}) = (1, 1, 1)$ is between $(1-\delta)\beta^3 m$ and $(1+\delta)\beta^3 m$ with high probability. Moreover, each block fragment corresponds to one honest block included in the main chain. This leads to the following theorem for the chain growth rate.

\begin{theorem} \label{theo:growth1}
The chain growth rate of pipelined HotStuff under the forking attack converges to $\beta^3$ with high probability as $m \to \infty$. 
\end{theorem}

\begin{proof}
By Lemma~\ref{lem:keylemma1} and the definition of $u_1(\alpha)$, we have $u_{1}(\alpha) = \lim_{m \rightarrow \infty}{\sum_{i=0}^{m-1} {Z_i} /m} \to \beta^3$.
\end{proof}

Note that when there does not exist the forking attack, the chain growth rate is $\beta$, for the probability that an honest node is elected as a leader is $\beta$. In other words, this theorem states that the forking attack reduces the chain growth rate
from $\beta$ to $\beta^3$.
For instance, if  $\beta = 2/3$, the chain growth rate is reduced from $2/3$ to $8/27$.

\begin{remark}
The chain growth rate is measured in terms of rounds. If it is measured in time, the forking attack can reduce the chain growth rate even more. This is because an adversarial leader can push the duration of its round close to the timeout value, which is usually much longer than the actual network delay for producing a certified block.
\end{remark}

\subsubsection{Chain Quality}
%With the analysis of the chain growth, it is easy to obtain the results of chain quality for pipelined HotStuff and LibraBFT. 
Recall that the adversary suffers no loss of blocks when launching the forking attack. That is,  every adversarial block can be kept in the main chain. 
Therefore, the adversary can produce $\alpha m$ adversarial blocks on expectation over $m$ rounds. 
This observation, together with Theorem~\ref{theo:growth1}, allows us to derive the following chain quality theorem for pipelined HotStuff. 
\begin{theorem}
The chain quality of pipelined HotStuff under the forking attack converges to $\frac{\beta^3}{\beta^3 - \beta + 1}$ with high probability as $m \to \infty$.
\end{theorem}

\begin{proof}
As $m \to \infty$, $\frac{B_a(m)}{m}$ (the number of adversarial blocks divided by $m$) converges to $\alpha$ by Lemma~3 in~\cite{HotStuffGit}, and $\frac{B_h(m)}{m}$ (the number of honest blocks in the main chain divided by $m$) converges to $\beta^3$ by Lemma~\ref{lem:keylemma1}.
Hence, the chain quality converges to $\frac{\beta^3}{\beta^3 -\beta + 1}$. Note that $\alpha = 1 - \beta$.
\end{proof}

For example, if $\alpha = 1/3$, the chain quality under the forking attack is $8/17$, whereas it should be $2/3$ without the forking attack.
%By contrast, in the optima case, the chain quality is $2/3$. 
%As said in Sec.~\ref{sec:model}, in the asynchronous network, the probability that the adversary is elected as a leader is bigger than $1/3$, thus the chain quality will be lower than $8/17$.  
If each block can bring its owner a reward, the adversary can obtain a fraction of rewards $\frac{\alpha}{\beta^3 -\beta + 1}$ by launching the forking attack, which is always higher than the deserved fraction $\alpha$, for $\beta^3 -\beta + 1 < 1$.
In other words, incentive compatibility of pipelined HotStuff cannot hold anymore under the forking attack. 

%Give more comparison.
\begin{figure}[t]
\centering
\includegraphics[width=2.2in]{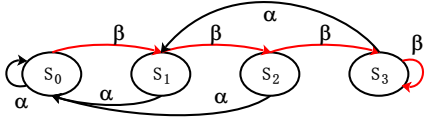}
\caption{\textbf{The state transition of the delay attack on pipelined HotStuff.}}
\vspace{-2mm}
\label{fig:systemTransit}
\end{figure}

\subsubsection{Latency} \label{subsec:latency}
Here, we compute the latency of honest blocks in the main chain. 
To achieve this goal, we need to track each honest block included in the main chain and obtain its associated delay to be committed. 
More precisely, the chance that an honest block is kept in the main chain and its delay is affected by the delay attack strategies, which further depend on the chain structure. 
For example, in Fig.~\ref{fig:delayAttack}, as shown in Case A, as there is no $3$-chain structure of the latest blocks, the adversary proposes no block and honest blocks $B_{k+2}$ and $B_{k+3}$ are kept in the main chain; however, they will be overridden in Case B. 
%(See detailed attack strategies in Sec.~\ref{subsec:attacks}.)
%The different attack strategies further affect honest blocks' chance to be included in the main chain and the delay to be committed.  
Therefore, we need to track the previous block structure of every honest block. To this end, we define four states as follows:
%We are now ready to introduce states:
\begin{itemize}[leftmargin=*]
    \item $S_{0}$: the state where the previous round is a timeout and no certified block is produced;

    \item $S_i$ for $i \in \left \{1, 2, 3 \right \} $: the state where there exists $i$ consecutive blocks that are not committed yet.
\end{itemize}
Under the delay attack described in Sec~\ref{subsec:attacks}, we can develop a Markov model of state transitions in Fig.~\ref{fig:systemTransit}.
Recall that $\alpha$ (respectively, $\beta$) is the probability that an honest (respectively, adversarial) leader proposes a new block. 
Each transition denotes a new round and there exists a designated leader.
An honest leader always proposes a new block that extends the newest certified block (denoted as the red line in Fig.~\ref{fig:systemTransit}).
%More precisely, we track whether the latest blocks are generated in consecutive rounds to enable block commitments.
%We now present the Markov model of the delay attack, as shown in Fig.~\ref{fig:systemTransit}.
%This is because, each honest block is affected by the delay attack strategies in the same way, and so go through the same random process to be committed.
%we find that the delay attack strategies affect every honest block in the main chain in the same way. 
%An honest leader always proposes a new block on the newest certified block.
%By contrast, an adversarial leader can either override previous honest blocks or cause a timeout. 
%(See the detailed attack strategies in Sec~\ref{subsec:attacks}.)
This Markov model allows us to track the previous chain structure for any new honest block as well as to obtain its chance to be included in the main chain and the associated delay. 
We have the following theorem on the delay of pipelined HotStuff. 
\begin{theorem}\label{theo:latency}
The latency of pipelined HotStuff under the delay attack converges to $\frac{\beta^7 + 3\beta^6 - 4\beta^5 + 2\beta^4 + \beta^3 - 2\beta^2 + \beta + 1}{2 \beta^7 - 2 \beta^6 + \beta^4}$ with high probability as $m \to \infty$.
\end{theorem}

\begin{proof}
First, by solving the above Markov model, we can obtain the steady-state distribution of each state as follows:
\begin{equation} \nonumber
\begin{split}
        \pi_0 &= \frac{(1+\beta)(1-\beta)^2}{\beta^3 - \beta^2 + 1}, \quad \pi_{1} = \frac{\beta (1-\beta)}{\beta^3 - \beta^2 + 1}, \\
 \pi_{2} &= \frac{ \beta^2 (1 - \beta)}{\beta^3 - \beta^2 + 1}, \quad \pi_{3} = \frac{\beta^3}{\beta^3 - \beta^2 + 1}.
\end{split}
\end{equation}
Next, we can analyze the latency of honest blocks in each state transition. 
In particular, we focus on the blocks that eventually end up in the main chain. 
%We detail the blocks on each event below.
\begin{itemize}[leftmargin=*]
\setlength{\itemsep}{0pt}
\setlength{\topsep}{0pt}
\setlength{\partopsep}{0pt} 
    \item \emph{Case $a$: $S_0 \xrightarrow{\beta} S_1$.} All proposed honest blocks will be kept in the main chain, and their average delay is $\frac{\beta^3 + \beta + 1}{ \beta^4} $. 

    \item \emph{Case $b$: $S_{1} \xrightarrow{\beta} S_2$.} The honest blocks have $\alpha$ probability to be committed with an average delay $\frac{\beta^4 + 2  \beta^3 + \beta + 1}{ \beta^4}$, $\alpha \beta$ probability to be committed with an average delay $\frac{2 \beta^4 + \beta^3 + \beta + 1}{ \beta^4}$, and $\beta^2$ probability to be committed with an average delay $\frac{2 \beta^4 +  \beta^3 - \beta^2 + 1}{\beta^4}$.
    
    \item \emph{Case $c$: $S_{2} \xrightarrow{\beta}S_3$ and $S_{3} \xrightarrow{\beta}S_4$.} The honest blocks have $\beta^2$ probability to be committed with an average delay $\frac{2 \beta^4 +  \beta^3 - \beta^2 + 1}{\beta^4}$ and $\alpha \beta$ probability to be committed with an average delay $\frac{2 \beta^4 +  \beta^3 + \beta + 1}{\beta^4}$.
\end{itemize}
Due to space constraint, the detailed proofs of these cases are provided in Appendix~B1 of our technical report~\cite{HotStuffGit}. With these results, it is easy to obtain the latency as:
\iffalse
\begin{equation}
    \begin{aligned}
        u_{3}(\alpha) & =  ( \pi_{0} \frac{\beta^3 + \beta + 1}{ \beta^4} + \alpha \pi_{1} \frac{\beta^4 + 2  \beta^3 + \beta + 1}{ \beta^4}   \\ & \quad + \beta^2 (\pi_{1} + \pi_{2} + \pi_{3}) \frac{2 \beta^4 +  \beta^3 - \beta^2 + 1}{\beta^4}).
    \end{aligned}     
\end{equation}
\fi

\begin{equation}
    \begin{aligned}
        u_{3}(\alpha) = \frac{\beta^7 + 3\beta^6 - 4\beta^5 + 2\beta^4 + \beta^3 - 2\beta^2 + \beta + 1}{2 \beta^7 - 2 \beta^6 + \beta^4}.
    \end{aligned}     
\end{equation}
This completes the proof.
\end{proof}

%The theorem shows that the latency for one honest block to be committed is $\frac{1}{\beta} + \frac{1}{\beta^3} + \frac{1}{\beta^4}$. 
The theorem shows that when the adversary controls $1/3$ of Byzantine nodes (i.e., $\alpha = 1/3$ and $\beta = 2/3$), the average latency for committing one block under the delay attack is about $8.33$ rounds. 
By contrast, without the delay attack, a block is committed if the next three consecutive blocks extend it with a latency of $3$ rounds. 

\subsection{Performance Analysis of LibraBFT}\label{subsec:libraAnalysis}
We analyze the performance of LibraBFT under the attack strategies. 
In particular, we will evaluate the differences between pipelined HotStuff and LibraBFT. 
These differences made by LibraBFT aim to cut off the delay of relaying blocks' QCs or fasten the block commitment (see Sec.~\ref{sec:algorithm}). 
%Compared with pipelined HotStuff, LibraBFT has two subtle differences, which aim to optimize the performance (see Sec.~\ref{subsec:librabft}).
%We now evaluate the impact of these engineering optimizations under the attacks. 
%we analyze pipelined HotStuff in detail, whereas the analysis of LibraBFT is brief, for the analysis can be derived similarly. 

On the one hand, as the forking attack strategies are the same, LibraBFT has the same chain growth rate and chain quality as pipelined HotStuff.
In other words, these changes do not affect these two metrics.
On the other hand, we can develop a similar Markov model to analyze the delay attack in LibraBFT as shown in Fig.~\ref{fig:latencyLibra}. This allows us to obtain the following theorem on the latency for LibraBFT.

%In addition, when there is a $3$-direct chain, the adversarial adopts the strategy. 
%It is easy to see that this attack strategy can maximize the delay of the target block.
%More importantly, if each honest block in the main chain adopts the same delay attack strategy, the strategies are consistent with each other. 
%This implies that even we adopt the attack in the single block's target, it can actually maximize the delay of all honest blocks in the main chain. 
\begin{figure}[t]
\centering
\includegraphics[width=2.3in]{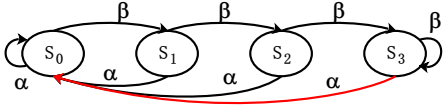}
\caption{\textbf{The state transition of the delay attack on LibraBFT.} The red line denotes a different action adopted by the adversarial leader in LibraBFT.}
\vspace{-2mm}
\label{fig:latencyLibra}
\end{figure}

%LibraBFT has different delay attack strategies. Specifically, when there are three consecutive blocks, the adversary in pipelined HotStuff chooses to override the last two blocks, while chooses to be faulty and to hide previous blocks' QC in LibraBFT.
%By following a similar analysis in pipelined HotStuff, we can present the Markov model of the attack in LibraBFT, as shown in Fig.~\ref{fig:latency2}.
%Therefore, the state $S_2$ transits to $S_{0}$ in LibraBFT, rather than $S_0$ in pipelined HotStuff.
%Furthermore, we can get the following latency theorem for LibraBFT. 

\begin{theorem}\label{theo:latency2}
The latency of LibraBFT under the delay attack converges to $\frac{\beta^7 + \beta + 1}{\beta^7 - \beta^6 + \beta^4}$ with high probability as $m \to \infty$. 
\end{theorem}

\begin{proof}
First, by solving the above Markov model, we can obtain the steady-state distribution of each state as follows:
\begin{equation} \nonumber
\begin{split}
        \pi_0 = 1-\beta, \quad \pi_{1} = \beta (1-\beta), \quad \pi_{2} = \beta^2 (1-\beta), \quad \pi_{3} = \beta^3.
\end{split}
\end{equation}
Next, we can analyze the latency of honest blocks in each state transition. We detail honest blocks on each event below.
\begin{itemize}[leftmargin=*]
\setlength{\itemsep}{0pt}
\setlength{\topsep}{0pt}
\setlength{\partopsep}{0pt} 
    \item \emph{Case $a$: $S_0 \xrightarrow{\beta} S_1$ and $S_{1} \xrightarrow{\beta} S_2$.} 
    All honest blocks produced after a previous timeout round or just one consecutive block will be kept in the main chain, and their average delay is $\frac{\beta^2 + \beta + 1}{\beta^4}$. 
    %Due to space limits, the detailed proofs are provided in Appendix.~\ref{subsec:LibraLatency}~\cite{HotStuffGit}.
    %The latency for honest blocks produced after a previous timeout round is $\frac{1}{\beta} + \frac{1}{\beta^3} + \frac{1}{\beta^4}$. For clarity, the detailed proofs are provided in Appendix.~\ref{subsec:HotLatency}.
    
    \item \emph{Case $b$: $S_{2} \xrightarrow{\beta}S_3$ and $S_{3} \xrightarrow{\beta}S_4$.} The honest blocks have $\beta^2$ probability to be committed with an average delay $\frac{2 \beta^4 + 1}{\beta^4}$ and $\alpha \beta$ probability to be committed with an average delay $\frac{2 \beta^4 +  \beta^3 + \beta^2 + \beta + 1}{\beta^4}$.
\end{itemize} 
Due to space constraint, the detailed proofs of these cases are provided in Appendix~B2 of our technical report~\cite{HotStuffGit}. With these results, it is easy to obtain the latency as:
\begin{equation}
    \begin{aligned}
        u_{3}(\alpha) = \frac{\beta^7 + \beta + 1}{\beta^7 - \beta^6 + \beta^4}.
    \end{aligned}     
\end{equation}
This completes the proof.
\end{proof}

The theorem shows that when $\beta = 2/3$, the average latency for committing one block under the delay attack is about $10.25$ rounds.
Compared with the delay in pipelined HotStuff, it suggests that the mechanism of sending votes to the next-round leader makes the system more vulnerable against the delay attack.

\section{Countermeasures} \label{sec:countermeasure}
%In this section, we present two countermeasures to thwart the aforementioned attacks and provide analyses of their improvements in performance.  

\subsection{Broadcasting QCs}
The first countermeasure is that current-round leaders broadcast QCs to all nodes (rather than just relaying QCs to next-round leaders)\footnote{An earlier version of LibraBFT has adopted this mechanism.}.
Broadcasting QCs can provide two benefits. 
First, when nodes receive QCs, they can update their $locked\_round$, which can effectively thwart the forking attack. 
For example, in Fig.~\ref{fig:fastcommit}, when honest nodes receive QC for block $B_{k+2}$, they can update $locked\_round$ from $k$ to $k+1$.
As a result, the adversarial leader of round $k+3$ can only override block $B_{k+2}$; however, without this mechanism, the adversarial leader can override both blocks $B_{k+1}$ and $B_{k+2}$.
Second, broadcasting QCs can fasten block commitments.
Specifically, when nodes observe a $3$-direct chain of $B_k$, $B_{k+1}$, and $B_{k+2}$, as well as  block $B_{k+2}$'s QC, they can commit block $B_{k}$, as shown in Fig.~\ref{fig:fastcommit}. 
By contrast, in pipelined HotStuff, without broadcasting QCs, nodes need to wait until the subsequent block carrying block $B_{k+2}$'s QC is received and then commit block $B_{k}$.
As said in Sec.~\ref{subsec:latency}, this enables an adversarial leader at round $k+3$ to hide $B_{k+2}$'s QC and ruin the $3$-direct chain to increase block delay. 
%Intuitively, it is easy to see that broadcasting QC can accelerate block commitment. 
Additionally, even in the ideal case, broadcasting QC can accelerate the block commitment; the latency for committing one block is reduced to two rounds. In the following, we provide a formal analysis of the improvement in chain growth rate, chain quality, and delay.
%Our analysis shows that broadcasting QC can effectively thwart delay attack.  

\begin{figure}[t]
\centering
\includegraphics[width=2in]{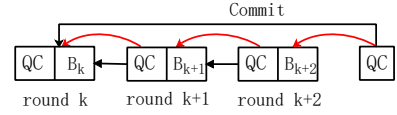}
\caption{\textbf{The committing rule with broadcasting QCs in Pipelined HotStuff.}}
\vspace{-2mm}
\label{fig:fastcommit}
\end{figure}

\subsubsection{Chain Growth Rate and Chain Quality}
With broadcasting QCs, an honest block proposed at round $i$ can only be overridden by an adversarial leader at round $(i+1)$. In other words, an honest block can be kept in the main chain if and only if the subsequent leader is an honest leader. 
Let $Y_i$ denote an indicator random variable, which equals to one if $\{ X_i = 1, X_{i+1} = 1\}$ and equals to zero otherwise. Next, let $Y = \sum_{i=1}^{m}Y_i$.
By following our previous analysis, we can bound the value of $Y$:

\begin{lemma}
For $m$ consecutive rounds, the number of block fragments $(X_i, X_{i+1}) = (1,1)$ has the following Chernoff-type bound: For $0 < \delta < 1$,
\begin{equation}
   \Pr( | Y - \beta^2 m| > \delta \beta^2 m ) < e^{-\Omega\left(\delta^2 \beta^2  m\right)}. 
\end{equation}
\end{lemma}
\begin{proof}
The analysis is very similar to that for Lemma~\ref{lem:keylemma1}.
\end{proof} 

With this lemma, we can easily derive the following theorems on chain growth rate and chain quality, respectively.

\begin{theorem} \label{theo:growth2}
The chain growth rate of pipelined HotStuff with broadcasting QCs under the forking attack converges to $\beta^2$ with high probability as $m \to \infty$. 
\end{theorem}

\begin{theorem}
The chain quality of pipelined HotStuff with broadcasting QCs under the forking attack converges to $\frac{\beta^2}{\beta^2 - \beta + 1}$ with high probability as $m \to \infty$.
\end{theorem}

These two theorems can be easily proved by following our previous analysis for pipelined HotStuff. Due to space constraints, we do not provide the proofs here. 
When $\beta = 2/3$, the chain growth rate improves by $1.5$x, while the chain quality improves by $1.2$x.

\subsubsection{Latency}
Following our previous analysis, we can develop a Markov model of the delay attack in Fig.~\ref{fig:latencyQC}. Note that as the adversarial leader always chooses to propose no blocks, all honest blocks can be kept in the main chain. 
Moreover, by solving the above Markov model, we have the following theorem.

\begin{theorem}\label{theo:latency4}
The latency of pipelined HotStuff with broadcasting QCs under the delay attack converges to $\frac{\beta + 1}{\beta^3}$ with high probability as $m \to \infty$. 
\end{theorem}

\begin{proof}
\iffalse
First, by solving the above Markov model, we can obtain the steady-state distribution of each state as follows:
\begin{equation} \nonumber
\begin{split}
        \pi_0 = 1-\beta, \quad \pi_{1} = \beta (1-\beta), \quad \pi_{2} = \beta^2.
\end{split}
\end{equation}
\fi
All produced honest blocks will be kept in the main chain, and their average delay is $\frac{\beta + 1}{ \beta^3}$. Due to space constraint, the detailed proofs of these cases are provided in Appendix~B3 of our technical report~\cite{HotStuffGit}.
\end{proof}

This theorem shows that when $\beta = 2/3$, the average latency for committing one block under the delay attack is $5.63$ rounds.
It suggests that broadcasting QCs can reduce the average block latency by almost $3$ rounds. 
Note that broadcasting QCs also brings additional delay.
Therefore, it is a design tradeoff, which should be evaluated in real settings in order to decide whether to adopt it.

\begin{figure}[t]
\centering
\includegraphics[width=1.8in]{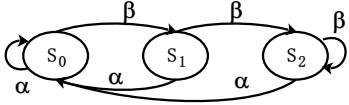}
\caption{\textbf{The state transition of the delay attack on pipelined HotStuff with broadcasting QCs.}}
\vspace{-2mm}
\label{fig:latencyQC}
\end{figure}

\subsection{Longest Chain Rule}
Our second countermeasure is changing the block proposing rule. 
In pipelined HotStuff, an honest node always extends the newest certified block. 
This deterministic block proposing rule enables the adversary to override at most two previous honest blocks by a higher certified block without any loss (i.e., decreasing chain growth rate and chain quality) and to break the $3$-direct chain (i.e., increasing latency).
%all honest nodes will accept and propose new blocks on this block. 
Therefore, we suggest that nodes can choose to  extend the longest certified blockchain
%\footnote{In LibraBFT, it sometimes says that nodes extend the longest chain. This is because in the ideal case, the newest certified block is the head of the longest chain. %However, nodes still extend the newest certified block.}. 
In particular, when there are two forking branches with the same length, they randomly choose one to extend. This randomized block proposing rule is inspired by the longest chain rule in NC. 
A detailed analysis of this countermeasure will be provided in a journal version of this work.

%This section provides a simple analysis of selfish mining in Crystal, in which the delay $\Delta$ is assumed to approach zero. Thus, no blocks are mined when other blocks are being transmitted. This assumption is widely used in the selfish mining analysis of Bitcoin~\cite{eyal2014majority, sapirshtein2016optimal,nayak2016stubborn}\footnote{The assumption can be justified because the block generation interval in Bitcoin (i.e., $10$ minutes on average) is significantly larger than the block propagation delay (e.g., smaller than $10$ seconds \cite{Decker2013}). Crystal adopts a similar setting.}. Later, we will relax this assumption and consider the impact of the network delay in our evaluations. 

%Unlike NC, when the attacker in Crystal mines a block in advance and decides to withhold it, it cannot mine the next block on this block without a certificate (see Lemma~\ref{lem:majority}). 
%Obviously, the quorum certificates needed in the mining process prohibit the attacker to secretly increase its block lead to more than one block. 
%In other words, to launch the attack, the attacker has to stop mining, wait for honest nodes to produce a block, and then publish its block to create a fork on purpose. 
%The rest honest nodes and the attacker will mine on top of the attacker's branch. The pseudocode for the selfish mining algorithm can be found in Appendix~\ref{appen:selfish}.

\section{Evaluation} \label{sec:evaluation}
We implement a proof-of-concept of pipelined HotStuff to evaluate its performance in terms of chain growth rate, chain quality, and latency under the forking and delay attacks. 
%Experiments show the performance of pipelined HotStuff under different situations.

\subsection{Testnet Setup}
We consider a system of $16$ nodes, and the number of Byzantine nodes is up to $5$. 
For simplicity, nodes are set to have synchronized clocks, and so the protocol proceeds in synchronized rounds\footnote{In a partially synchronous network, nodes can establish a synchronized clock as long as they have clocks with bounded drift~\cite{dwork1988consensus}.}.
We build a full-fledged implementation of pipelined HotStuff using Golang (around $3,600$ LoC). 
We run the simulation on a late 2013 Apple MacBook Pro, 2.7GHz Intel Core i7.
%In addition, to measure the latency in time, we evaluate the block latency under the delay attack on 16 Google Cloud N1 high-CPU instances, each of which has 4 vCPUs supported by Intel Xeon E5.
%We run each node on a single docker container.
%The propagation delay between any two nodes follows the exponential distribution with the expected delay of $30$ ms. 
%The computation delay like verifying votes and QCs are ignored in the simulations. 
In our experiments, the adversary runs the attack strategies in Sec.~\ref{subsec:attacks}. 
Our simulation results are based on an average of $10$ runs, where each run generates 100,000 blocks.
%The implementation program was written in Go and follows an event-driven fashion.

\iffalse
\begin{figure}[t]
\centering
\includegraphics[width=1.8in]{images/latencyPlot.eps}
\caption{The average latency of pipelined HotStuff and LibraBFT under the delay attack.}
\vspace{-2mm}
\label{fig:latencyPlot}
\end{figure}
\fi

\begin{figure}[t]
\centering      
\subfigure[The chain growth rate under the forking attack.]{      
\begin{minipage}{4cm}
\centering                                                          \includegraphics[width=1.8in]{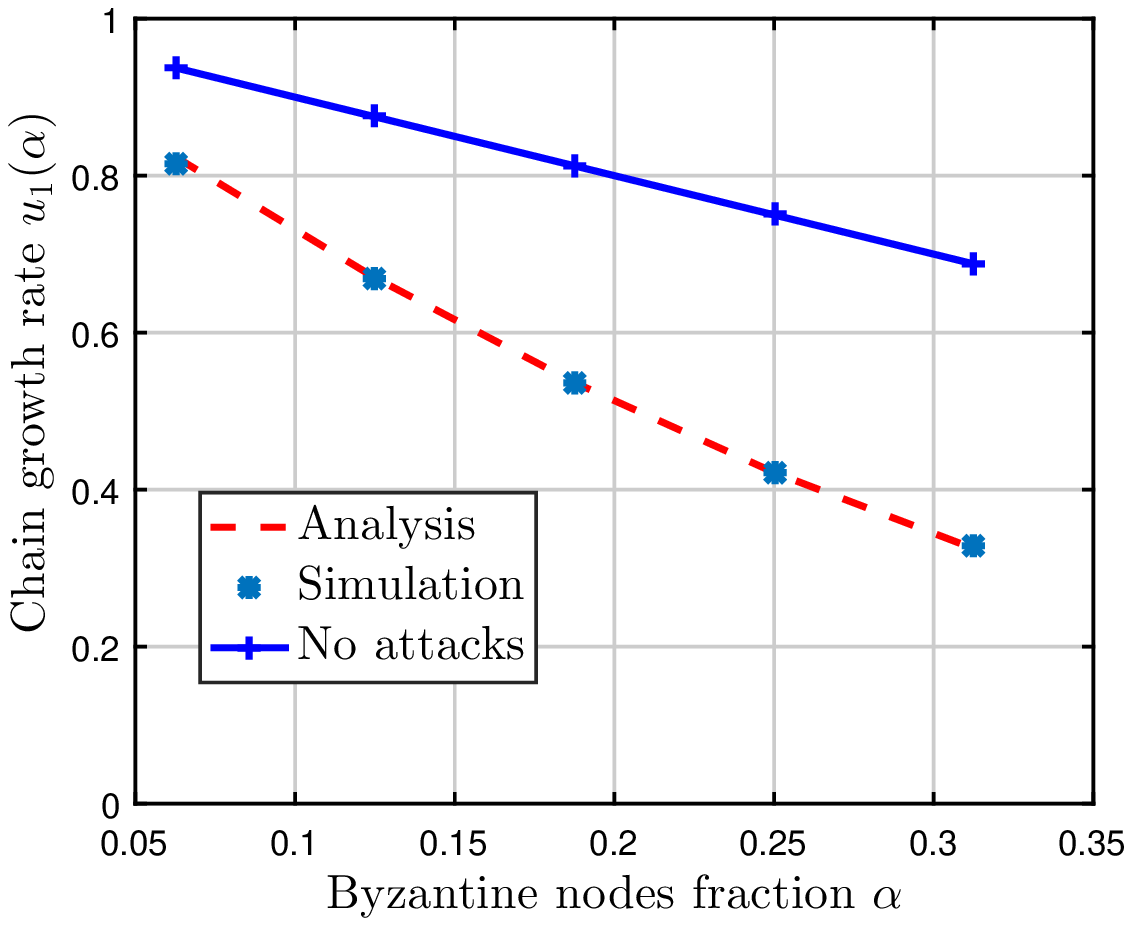}   
\label{fig:growthQualityA}   
\end{minipage}
}
\subfigure[The chain quality under the forking attack.]{    
\begin{minipage}{4cm}
\centering                                                          \includegraphics[width=1.85in]{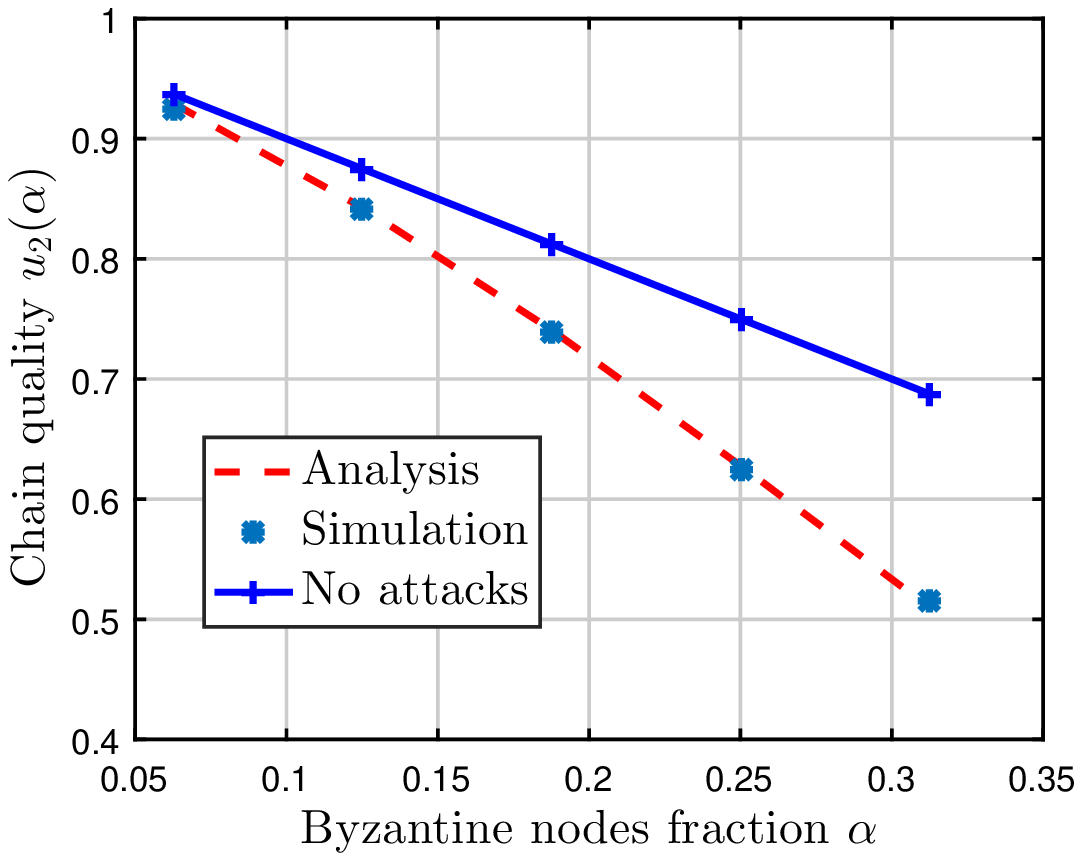}   
\label{fig:growthQualityB} 
\end{minipage}
}
\subfigure[The average latency of honest blocks under the delay attack.]{    
\begin{minipage}{8cm}
\centering                                                          \includegraphics[width=1.9in]{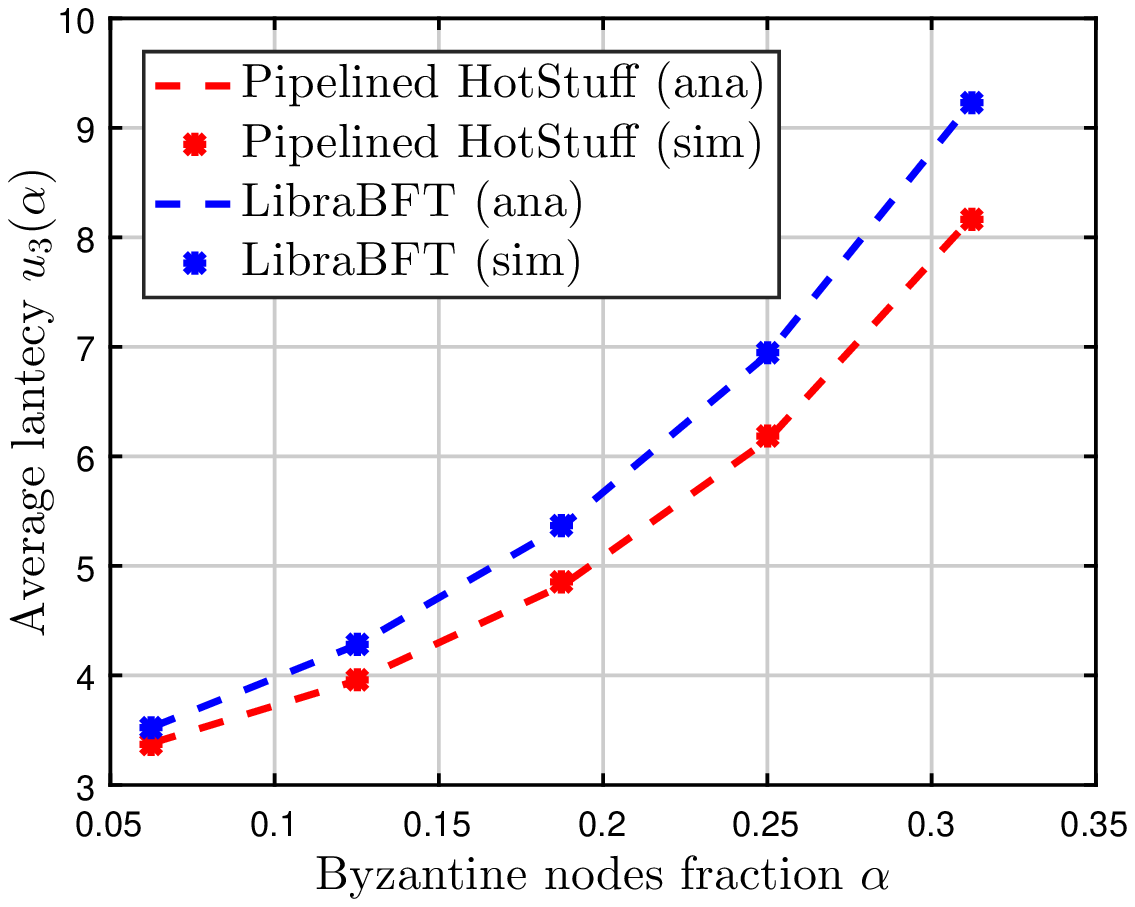}   
\label{fig:growthQualityC}   
\end{minipage}
}
\caption{The performance of pipelined HotStuff as well as LibraBFT under the forking and delay attacks.}  
\label{fig:growthQuality}                                                        
\end{figure}

\subsection{Pipelined HotStuff and LibraBFT}
We evaluate the performance of pipelined HotStuff as well as LibraBFT through extensive experiments. Since LibraBFT has the same chain growth rate and chain quality as pipelined HotStuff, these two performance metrics are only given for pipelined HotStuff.
%In particular, we change the number of Byzantine nodes to observe the impact of different fractions of Byzantine nodes on the performance.
%The experiment results well match our analysis.

\subsubsection{Chain Growth Rate}
Fig.~\ref{fig:growthQualityA} shows the chain growth rate of pipelined HotStuff with different fractions of Byzantine nodes. First, we observe that the simulation results well match the analysis results. Second, the results show that as the fraction $\alpha$ of Byzantine nodes increases, the gap between the chain growth rates with and without the forking attack also increases. When $\alpha$ is close to $0.3$, the chain growth rate under the attack can drops to almost half of that without attacks. 

\subsubsection{Chain Quality}
Fig.~\ref{fig:growthQualityB} shows the chain quality of pipelined HotStuff with different fractions of Byzantine nodes.
First, the evaluation results match our previous analysis.
Second, The results show that by the forking attack, the adversary can lower the chain quality and obtain a higher fraction of blocks in the main chain than what it deserves. 
If each block in the main chain brings to its owner a reward, this implies that the adversary can always gain a higher fraction of rewards.
In other words, the incentive compatibility cannot be held anymore under the attack.
For example, when $\alpha$ is close to $0.3$, the chain quality drops to $0.51$. 
This implies that $0.3$ of Byzantine nodes can produce almost half of the blocks in the main chain (and obtain half of the rewards). 

\subsubsection{Latency}
Fig.~\ref{fig:growthQualityC} shows the average latency of honest blocks in the main chain in pipelined Hotstuff and LibraBFT. 
First, the evaluation results, once again, validate our analysis. 
Second, the results show that as the fraction $\alpha$ of Byzantine nodes increases, both the latency in pipelined Hotstuff and LibraBFT increase. 
In addition, the latency in LibraBFT is larger than that in pipelined HotStuff, which implies that the engineering optimizations adopted by LibraBFT may make it more vulnerable to the delay attack. 
Note that in both pipelined HotStuff and LibraBFT, the latency for committing one block is three rounds without attacks. 
Therefore, when $\alpha$ is close to $0.3$, the average latency of honest blocks under the delay attacks is almost $3$x of that without attacks.

\subsection{Countermeasures}
We evaluate the performance of pipelined HotStuff with broadcasting QCs.  
Fig.~\ref{fig:growthQualityQCA} and~\ref{fig:growthQualityQCB} show that as the fraction $\alpha$ of Byzantine nodes increases, the gap between the chain growth rates (and chain qualities) of pipelined HotStuff with and without broadcasting QCs also increases.
This implies that the higher $\alpha$ is, the higher performance improvement that broadcasting QCs brings.
Fig.~\ref{fig:growthQualityQCC} shows that the latency of pipelined HotStuff with broadcasting QCs is at least one round shorter than the original HotStuff. 
In addition, when $\alpha$ is close to $0.3$, the average latency of honest blocks can drop by almost $3$ rounds.
As previously explained, broadcasting QCs also brings additional delay. Therefore, it is a design tradeoff that should be evaluated in \emph{real settings}. 
Finally, we would like to point out that the longest chain rule can also significantly enhance the performance of pipelined HotStuff. 
Moreover, the longest chain rule brings no additional overhead, and it can be combined with broadcasting QCs. We will present these results in a journal version of this work.
%More details will be shown in our later journal extension. 

\begin{figure}[t]
\centering      
\subfigure[The chain growth rate under the forking attack.]{      
\begin{minipage}{4cm}
\centering                                                          \includegraphics[width=1.8in]{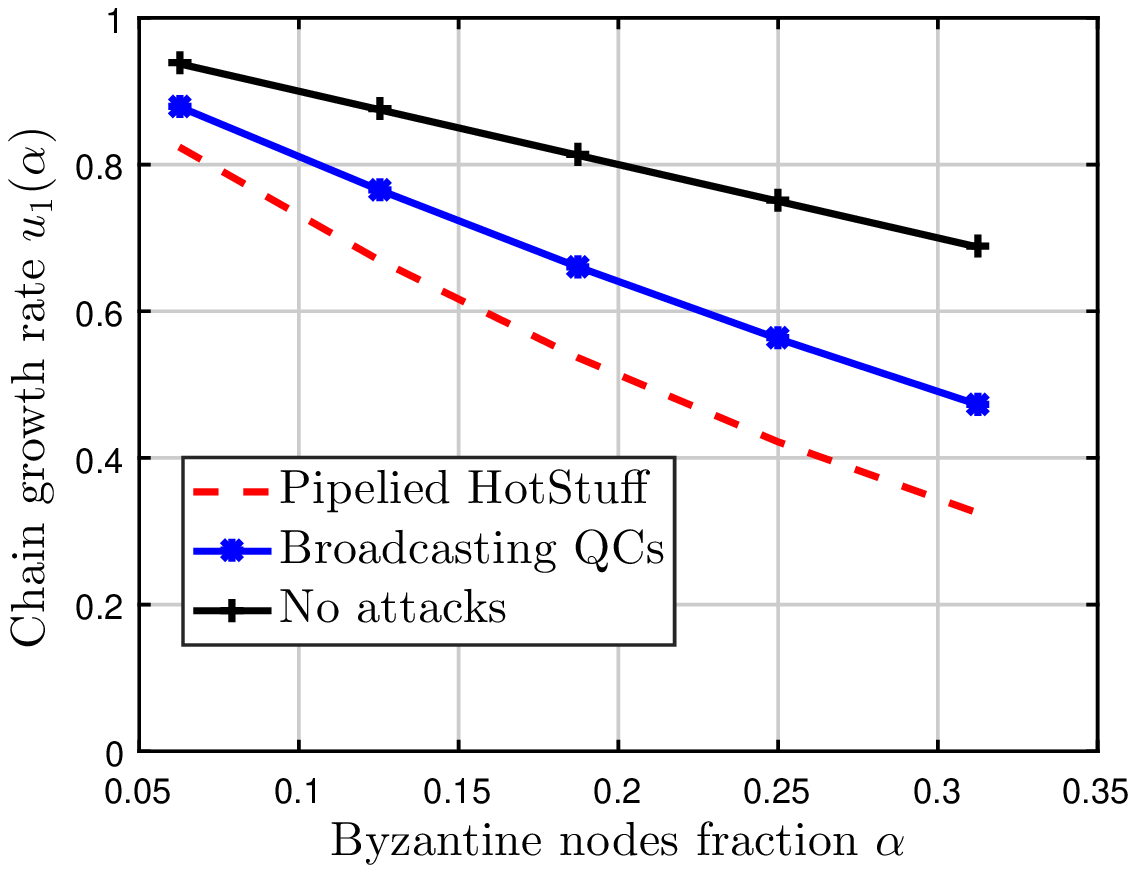}   
\label{fig:growthQualityQCA}   
\end{minipage}
}
\subfigure[The chain quality under the forking attack.]{    
\begin{minipage}{4cm}
\centering                                                          \includegraphics[width=1.8in]{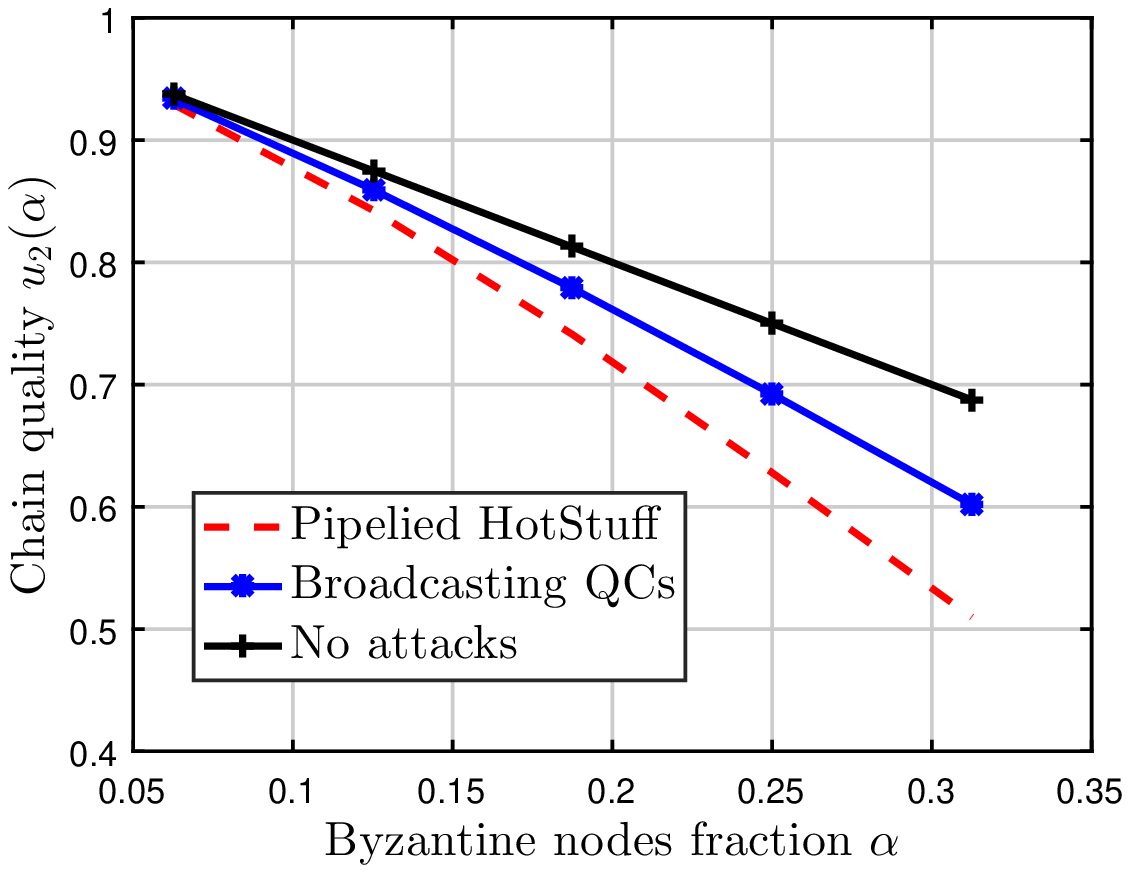}   
\label{fig:growthQualityQCB} 
\end{minipage}
}
\subfigure[The average latency under the delay attack.]{    
\begin{minipage}{8cm}
\centering                                                          \includegraphics[width=1.8in]{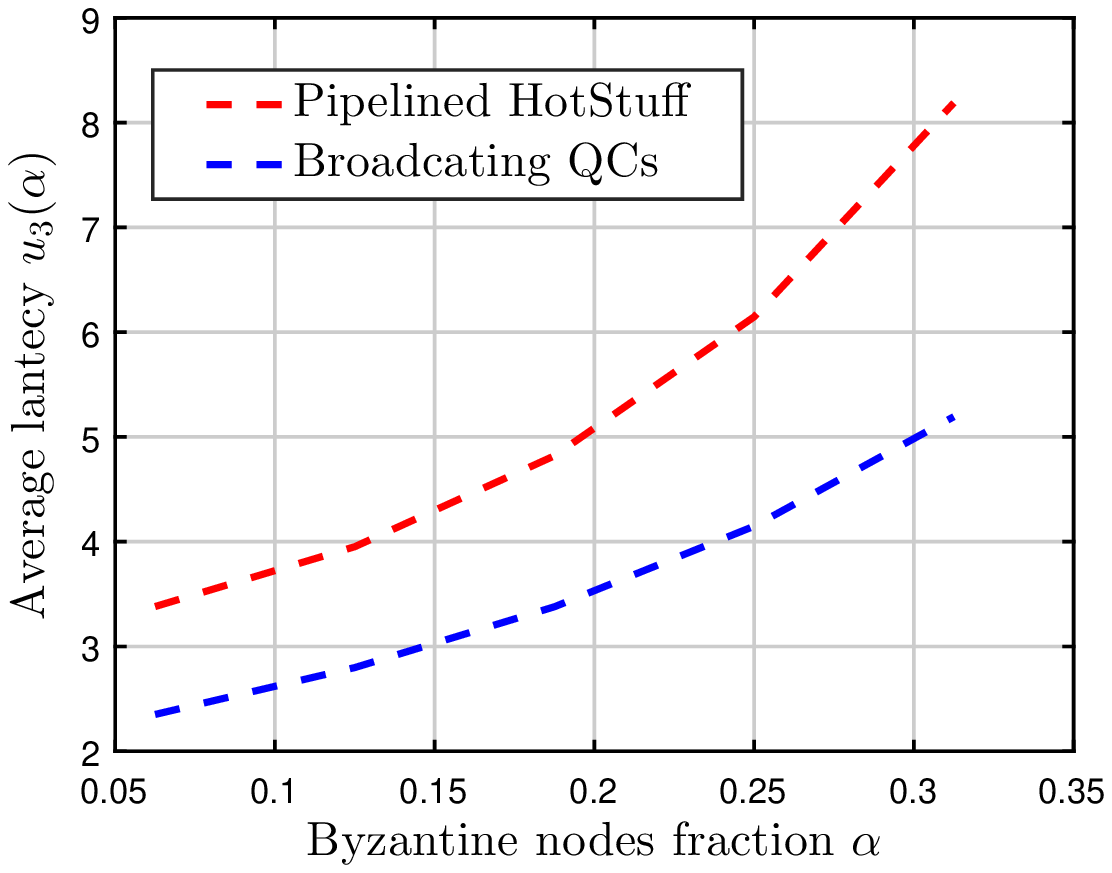}   
\label{fig:growthQualityQCC}   
\end{minipage}
}
\caption{The performance of pipelined HotStuff under the forking and delay attacks with and without broadcasting QCs.}  
\label{fig:growthQualityQC}                                                        
\end{figure}

\section{Related Work} \label{sec:related}
Reaching consensus in face of Byzantine failures was formulated as the Byzantine agreement problem by Lamport \etal \cite{lamport82byzantinegeneral}, and has been studied for several decades. 
Various BFT consensus protocols such as PBFT~\cite{pbft1999}, Zyzzyva~\cite{kotla2007zyzzyva}, and Q/U~\cite{QU} have been proposed.
However, these classical BFT protocols suffer from poor scalability, notorious complexity and leader fairness issues (see Sec.~\ref{sec:intro}) and are hard to be used in large-scale blockchains.
To address these issues, several state-of-the-art BFT protocols~\cite{buchman2016tendermint, casper, HotStuffYin2019, chan2018pala} are proposed for building large-scale blockchains. 

\noindent \textbf{Tendermint.} Tendermint~\cite{buchman2016tendermint} features a continuous leader rotation (also called the democracy-favoring leader rotation~\cite{chan2018pala}) based on PBFT protocol.
Specifically, Tendermint embeds the round-change mechanism into the common-case pattern, and the leader is re-elected from all the nodes by some desired policy after every block, resulting in better leadership fairness.
%As a result, each node in the system can be frequently elected as a leader and win some rewards, resulting in better leadership fairness and incentive mechanism design. 
%Note that a similar method is previously used in some BFT protocols, but it aims to mitigate the PRE-PREPARE attack~\cite{Clement09, Spin2009, bar2005, Prime2011}. 

\noindent \textbf{Casper FFG}. Buterin and Griffith~\cite{casper} proposed a protocol called Casper FFG, which works as an overlay atop NC to provide ``finality gadget''.
Casper FFG applies an elegant pipelining idea to the classical BFT protocol, i.e., if each block required two rounds of voting, one can piggyback the second round on the next block’s voting. 
This pipelining idea enables the system to have one identical round (rather than multiple rounds with different functionalities and names\footnote{In PBFT, for committing one proposal, there are two phases: prepare and commit phases, and each phase has different functionalities.}), and so significantly simplifies the protocol design. 
%For example, in PBFT, there are two different rounds, which are the “prepare” and “commit” round, whereas, in Casper FFG, all the rounds have the same functions. 

\noindent \textbf{Pala and Streamlet.} Pala~\cite{chan2018pala} is a simple BFT consensus protocol that also adopts the pipelining idea. However, for high throughput, it uses a stability-favoring leader rotation policy. Based on this work, Chan \etal~\cite{Chan2020StreamletTS} proposed Streamlet, which further simplifies the voting rule. Streamlet aims to provide a unified, simple protocol for both teaching and implementation.

\noindent \textbf{HotStuff.} HotStuff proposed by Yin \etal~\cite{HotStuffYin2019} creatively adopts a three-phase commit rule (rather than the two-phase commit rule used in Casper FFG, Pala, and Streamlet)  to enable the protocol to reach consensus at the pace of actual network delay. In addition, HotStuff adopts the threshold signature to realize linear message complexity, and can also be pipelined into a practical protocol for building large-scale blockchains.

\noindent \textbf{Fast-HotStuff.} Fast-HotStuff \cite{fastHotStuff} has lower latency compared to the HotStuff and is resilient to a forking attack. But unlike HotStuff, Fast-HotStuff adds a small overhead to the block during an unhappy path (when the primary fails).

\section{Conclusion} \label{sec:conclusion}
The state-of-the-art pipelined HotStuff not only provides linear message complexity and responsiveness but also is efficient for building large-scale blockchains. 
Thus, pipelined HotStuff has been adopted in many blockchain projects such as Libra, Flow, and Cypherium. 
In this paper, we propose a multi-metric evaluation framework including chain growth rate, chain quality, and latency. 
We also propose two attacks, namely the forking attack and delay attack, and systematically study the impacts of these two attacks on the performance of pipelined HotStuff. 
Also, we leverage the framework to evaluate some engineering designs in LibraBFT. 
Finally, we propose some countermeasures to enhance the performance of pipelined HotStuff against these attacks. 
We hope that our framework can contribute to proposing new variants of HotStuff as well as making HotStuff more understandable for developers and practitioners in terms of performance.
\bibliographystyle{IEEEtran}
\bibliography{short}

% Generated by IEEEtran.bst, version: 1.14 (2015/08/26)
\begin{thebibliography}{10}
\providecommand{\url}[1]{#1}
\csname url@samestyle\endcsname
\providecommand{\newblock}{\relax}
\providecommand{\bibinfo}[2]{#2}
\providecommand{\BIBentrySTDinterwordspacing}{\spaceskip=0pt\relax}
\providecommand{\BIBentryALTinterwordstretchfactor}{4}
\providecommand{\BIBentryALTinterwordspacing}{\spaceskip=\fontdimen2\font plus
\BIBentryALTinterwordstretchfactor\fontdimen3\font minus
  \fontdimen4\font\relax}
\providecommand{\BIBforeignlanguage}[2]{{%
\expandafter\ifx\csname l@#1\endcsname\relax
\typeout{** WARNING: IEEEtran.bst: No hyphenation pattern has been}%
\typeout{** loaded for the language `#1'. Using the pattern for}%
\typeout{** the default language instead.}%
\else
\language=\csname l@#1\endcsname
\fi
#2}}
\providecommand{\BIBdecl}{\relax}
\BIBdecl

\bibitem{nakamoto2012bitcoin}
S.~Nakamoto, ``Bitcoin: A peer-to-peer electronic cash system,'' \emph{Working
  Paper}, 2008.

\bibitem{garay2015}
J.~Garay, A.~Kiayias, and N.~Leonardos, ``The {B}itcoin backbone protocol:
  {A}nalysis and applications,'' in \emph{Advances in Cryptology - EUROCRYPT
  2015}.\hskip 1em plus 0.5em minus 0.4em\relax Berlin Heidelberg: Springer,
  2015, pp. 281--310.

\bibitem{Pass2017}
R.~Pass, L.~Seeman, and A.~Shelat, ``Analysis of the blockchain protocol in
  asynchronous networks,'' in \emph{Advances in Cryptology -- EUROCRYPT
  2017}.\hskip 1em plus 0.5em minus 0.4em\relax Cham: Springer, 2017, pp.
  643--673.

\bibitem{ghost}
Y.~Sompolinsky and A.~Zohar, ``Secure high-rate transaction processing in
  {B}itcoin,'' in \emph{Financial Cryptography and Data Security}.\hskip 1em
  plus 0.5em minus 0.4em\relax Berlin, Heidelberg: Springer, 2015, pp.
  507--527.

\bibitem{lamport1980}
M.~Pease, R.~Shostak, and L.~Lamport, ``Reaching agreement in the presence of
  faults,'' \emph{J. ACM}, vol.~27, no.~2, p. 228–234, Apr. 1980.

\bibitem{pbft1999}
M.~Castro and B.~Liskov, ``Practical {B}yzantine fault tolerance,'' in
  \emph{Proceedings of the Third Symposium on Operating Systems Design and
  Implementation}, ser. OSDI '99.\hskip 1em plus 0.5em minus 0.4em\relax
  Berkeley, CA, USA: USENIX Association, 1999, pp. 173--186.

\bibitem{Bessani2014}
A.~{Bessani}, J.~{Sousa}, and E.~E.~P. {Alchieri}, ``State machine replication
  for the masses with {BFT-SMART},'' in \emph{2014 44th Annual IEEE/IFIP
  International Conference on Dependable Systems and Networks}, 2014, pp.
  355--362.

\bibitem{Sousa2015}
J.~{Sousa} and A.~{Bessani}, ``Separating the wheat from the chaff: An
  empirical design for geo-replicated state machines,'' in \emph{2015 IEEE 34th
  Symposium on Reliable Distributed Systems (SRDS)}, 2015, pp. 146--155.

\bibitem{QU}
M.~Abd-El-Malek, G.~R. Ganger, G.~R. Goodson, M.~K. Reiter, and J.~J. Wylie,
  ``Fault-scalable {B}yzantine fault-tolerant services,'' \emph{SIGOPS Oper.
  Syst. Rev.}, vol.~39, no.~5, p. 59–74, Oct. 2005.

\bibitem{700BFT}
R.~Guerraoui, N.~Kne\v{z}evi\'{c}, V.~Qu\'{e}ma, and M.~Vukoli\'{c}, ``The next
  700 {BFT} protocols,'' in \emph{Proceedings of the 5th European Conference on
  Computer Systems}, ser. EuroSys ’10.\hskip 1em plus 0.5em minus 0.4em\relax
  New York, NY, USA: Association for Computing Machinery, 2010, p. 363–376.

\bibitem{mickens2014saddest}
J.~Mickens, ``The saddest moment,'' \emph{Login Usenix Mag}, vol.~39, no.~3,
  pp. 52--54, 2014.

\bibitem{QuestMarko}
M.~Vukoli{\'{c}}, ``The quest for scalable blockchain {F}abric:
  {P}roof-of-{W}ork vs. {BFT} replication,'' in \emph{Open Problems in Network
  Security}.\hskip 1em plus 0.5em minus 0.4em\relax Cham: Springer, 2016, pp.
  112--125.

\bibitem{chan2018pala}
T.-H.~H. Chan, R.~Pass, and E.~Shi, ``Pala: A simple partially synchronous
  blockchain.'' 2018.

\bibitem{HotStuffYin2019}
M.~Yin, D.~Malkhi, M.~K. Reiter, G.~G. Gueta, and I.~Abraham, ``Hot{S}tuff:
  {BFT} consensus with linearity and responsiveness,'' pp. 347--356, 2019.

\bibitem{banostate}
\BIBentryALTinterwordspacing
S.~Bano, M.~Baudet, A.~Ching, A.~Chursin, G.~Danezis, F.~Garillot, Z.~Li,
  D.~Malkhi, O.~Naor, D.~Perelman \emph{et~al.}, ``State machine replication in
  the {L}ibra blockchain,'' May 2020. [Online]. Available:
  \url{https://developers.libra.org/docs/state-machine-replication-paper.}
\BIBentrySTDinterwordspacing

\bibitem{Hentschel2020FlowSC}
A.~Hentschel, Y.~Hassanzadeh-Nazarabadi, R.~Seraj, D.~Shirley, and L.~Lafrance,
  ``Flow: Separating consensus and compute--block formation and execution,''
  2020.

\bibitem{cypherium}
\BIBentryALTinterwordspacing
Y.~Guo, Q.~Yang, H.~Zhou, W.~Lu, and S.~Zeng, ``Syetem and methods for
  selection and utilizing a committee of validator nodes in a distributed
  system,'' {Cypherium Blockchain}, Feb 2020, patent. [Online]. Available:
  \url{https://github.com/cypherium/patent}
\BIBentrySTDinterwordspacing

\bibitem{dwork1988consensus}
C.~Dwork, N.~Lynch, and L.~Stockmeyer, ``Consensus in the presence of partial
  synchrony,'' \emph{Journal of the ACM (JACM)}, vol.~35, no.~2, pp. 288--323,
  1988.

\bibitem{raft2014}
D.~Ongaro and J.~Ousterhout, ``In search of an understandable consensus
  algorithm,'' in \emph{2014 {USENIX} Annual Technical Conference ({USENIX}
  {ATC} 14)}.\hskip 1em plus 0.5em minus 0.4em\relax Philadelphia, PA: {USENIX}
  Association, Jun. 2014, pp. 305--319.

\bibitem{ZooKeepper}
P.~Hunt, M.~Konar, F.~P. Junqueira, and B.~Reed, ``Zookeeper: Wait-free
  coordination for internet-scale systems,'' in \emph{Proceedings of the 2010
  USENIX Conference on USENIX Annual Technical Conference}, ser.
  USENIXATC'10.\hskip 1em plus 0.5em minus 0.4em\relax USA: USENIX Association,
  2010, p.~11.

\bibitem{spiegelman2019ace}
A.~Spiegelman and A.~Rinberg, ``{ACE}: Abstract consensus encapsulation for
  liveness boosting of state machine replication,'' 2019.

\bibitem{bano2020twins}
S.~Bano, A.~Sonnino, A.~Chursin, D.~Perelman, and D.~Malkhi, ``Twins:
  White-glove approach for {BFT} testing,'' 2020.

\bibitem{hotstuffLatest}
M.~Yin, D.~Malkhi, M.~K. Reiter, G.~G. Gueta, and I.~Abraham, ``Hotstuff: Bft
  consensus in the lens of blockchain,'' 2018.

\bibitem{Gilad2017}
Y.~Gilad, R.~Hemo, S.~Micali, G.~Vlachos, and N.~Zeldovich, ``Algorand: Scaling
  {B}yzantine agreements for cryptocurrencies,'' in \emph{Proceedings of the
  26th Symposium on Operating Systems Principles}, ser. SOSP '17.\hskip 1em
  plus 0.5em minus 0.4em\relax New York, NY, USA: ACM, 2017, pp. 51--68.

\bibitem{Pass2017fruit}
R.~Pass and E.~Shi, ``Fruitchains: A fair blockchain,'' in \emph{Proceedings of
  the ACM Symposium on Principles of Distributed Computing}, ser. PODC
  '17.\hskip 1em plus 0.5em minus 0.4em\relax New York, NY, USA: ACM, 2017, pp.
  315--324.

\bibitem{eyal2014majority}
I.~Eyal and E.~G. Sirer, ``Majority is not enough: {B}itcoin mining is
  vulnerable,'' \emph{Commun. ACM}, vol.~61, no.~7, pp. 95--102, Jun. 2018.

\bibitem{niu2019selfish}
J.~{Niu} and C.~{Feng}, ``Selfish mining in ethereum,'' in \emph{2019 IEEE 39th
  International Conference on Distributed Computing Systems (ICDCS)}, Jul.
  2019, pp. 1306--1316.

\bibitem{niu2020incentive}
J.~Niu, Z.~Wang, F.~Gai, and C.~Feng, ``Incentive analysis of {B}itcoin-{NG},
  revisited,'' in \emph{Performance Evaluation: An International Journal}, vol.
  144.\hskip 1em plus 0.5em minus 0.4em\relax Elsevier, 2020, p. 102144.

\bibitem{eclipse2015}
E.~Heilman, A.~Kendler, A.~Zohar, and S.~Goldberg, ``Eclipse attacks on
  {B}itcoin’s peer-to-peer network,'' in \emph{Proceedings of the 24th USENIX
  Conference on Security Symposium}.\hskip 1em plus 0.5em minus 0.4em\relax
  USA: USENIX Association, 2015, p. 129–144.

\bibitem{HotStuffGit}
J.~Niu, F.~Gai, M.~M. Jalalzai, and C.~Feng, ``On the performance of pipelined
  hotstuff,'' 2021, preprint available at
  https://github.com/infocom2021HotStuffReport/Report.

\bibitem{lamport82byzantinegeneral}
L.~Lamport, R.~E. Shostak, and M.~C. Pease, ``The {B}yzantine generals
  problem,'' \emph{{ACM} Trans. Program. Lang. Syst.}, vol.~4, no.~3, pp.
  382--401, 1982.

\bibitem{kotla2007zyzzyva}
R.~Kotla, L.~Alvisi, M.~Dahlin, A.~Clement, and E.~Wong, ``Zyzzyva: speculative
  {B}yzantine fault tolerance,'' \emph{ACM SIGOPS Operating Systems Review},
  vol.~41, no.~6, pp. 45--58, 2007.

\bibitem{buchman2016tendermint}
E.~Buchman, ``Tendermint: {B}yzantine fault tolerance in the age of
  blockchains.'' M. Eng. thesis, The University of Guelph, Ontario, Canada,
  Jun. 2016.

\bibitem{casper}
V.~Buterin and V.~Griffith, ``Casper the friendly finality gadget,'' 2017.

\bibitem{Chan2020StreamletTS}
B.~Y. Chan and E.~Shi, ``Streamlet: Textbook streamlined blockchains,'' 2020.

\bibitem{fastHotStuff}
J.~N. Mohammad M.~Jalalzai and C.~Feng, ``Fast-hotstuff: A fast and resilient
  hotstuff protocol,'' 2020.

\bibitem{niu2019analysis}
J.~Niu, C.~Feng, H.~Dau, Y.-C. Huang, and J.~Zhu, ``Analysis of {N}akamoto
  consensus, revisited,'' 2019.

\end{thebibliography}

\appendix
\subsection{Concentration Bounds}
We denote the probability of an event $E$ by $\Pr[E]$ and the expected value of a random variable $X$ by $\e{X}$. We will use the following bounds in our proofs.

\begin{lemma} [Chernoff bound for dependent random variables \cite{niu2019analysis}] \label{lem:key_step}
Let $T$ be a positive integer. Let $X^{(j)} = \sum_{i = 0}^{n-1} X_{j + iT}$ be the sum of $n$ independent indicator random variables and $\mu_j = \e{ X^{(j)} }$ for $j \in \{1, \ldots, T\}$. Let $X = X^{(1)} + \cdots + X^{(T)}$. Let $\mu = \min_j \{ \mu_j \}$. Then, for $0 < \delta < 1$, $\Pr\left[ X \le (1 - \delta) \mu T \right] \le e^{-\delta^2 \mu / 2}$ and $\Pr\left[ X \le (1 + \delta) \mu T \right] \le e^{-\delta^2 \mu / 3}$.
\end{lemma}

\subsection{Latency Analysis}
\subsubsection{Pipelined HotStuff} \label{subsec:HotLatency}
Based on the defined states in Sec.\ref{subsec:latency}, we introduce an additional state $S_x$, where the newest certified block together with its predecessor blocks form a $3$-chain structure.
By the aforementioned commit rule, when nodes observer a $3$-chain structure, the first block in the $3$-chain structure together with its predecessor blocks are all committed (see Sec.~\ref{subsec:hotstuffalgorithm}).
This further means when an honest block is kept in the main chain, and its descendant blocks first form a $3$-chain structure, it will be committed. 

By slightly modified the Markov model in Fig.~\ref{fig:systemTransit}, we can show the state transitions for a node to observe some new committed blocks (i.e., entering state $S_x$) in Fig.~\ref{fig:latencyHotStuff}. Furthermore, we can compute the expected rounds for a state $S_{i}$ ($i \in \{0, 1, 2, 3\}$) to first enter state $S_x$. To realize this, we introduce a new variable $X_i$, which denotes the number of rounds that starting from state $S_i$, the state first hits state $S_x$. In addition, we use 
$\e{X_i}$ to denote the expectation of $X_i$. With the state transitions, we can have the following equation.

\begin{figure}[t]
\centering
\includegraphics[width=2.3in]{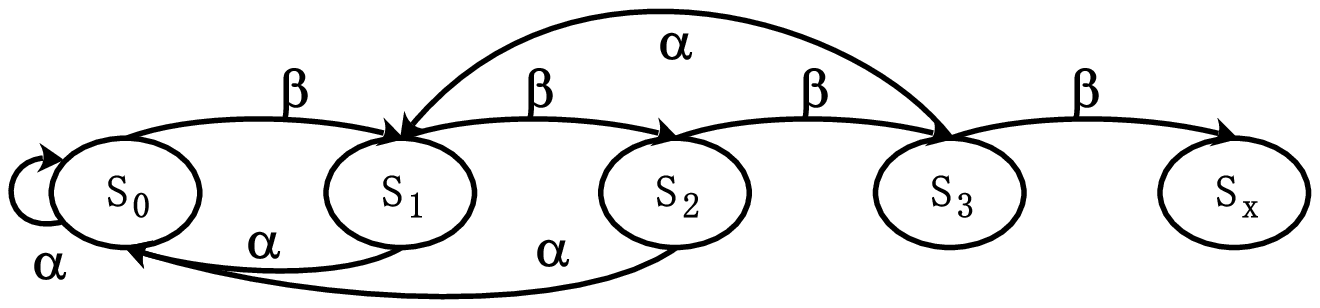}
\caption{\textbf{The state transition of the delay attack on pipelined HotStuff.}}
\vspace{-2mm}
\label{fig:latencyHotStuff}
\end{figure}

\begin{equation} \label{eq:delayHot}
\begin{split}
       \e{X_0} &= \alpha \e{X_{0}} + \beta \e{X_1} + 1 \\
       \e{X_1} &= \alpha \e{X_{0}} + \beta \e{X_2} + 1\\
       \e{X_2} &= \alpha \e{X_{0}} + \beta \e{X_3} + 1\\
       \e{X_3} &= \alpha \e{X_1} + 1\\
\end{split}
\end{equation}
By solving the equation~\eqref{eq:delayHot}, we can get:
\begin{equation} \nonumber
\centering
\begin{split}
        \e{X_0} = \frac{2\beta^3 + \beta + 1}{\beta^4}, \quad \e{X_1} = \frac{\beta^3 + \beta + 1}{\beta^4}, \\
 \e{X_2} = \frac{(1+\beta) (\beta^2 - \beta + 1)}{\beta^4}, \quad \e{X_3} = \frac{\beta^3 - \beta^2 + 1}{\beta^4}.\\
\end{split}
\end{equation}

Next, we can track honest blocks kept in the main chain and their associated delay. 
In addition, we refer to these tracked honest block as target blocks. 
As we said previously, the chance that an honest block is kept in the main chain and its delay is affected by the delay attack strategies, which further depend on the chain structure. 
Hence, according to its predecessor blocks structure, we can divide honest blocks into the following three cases:
\begin{itemize}[leftmargin=*]
\setlength{\itemsep}{0pt}
\setlength{\topsep}{0pt}
\setlength{\partopsep}{0pt} 
    \item \emph{Case $a$: $S_0 \xrightarrow{\beta} S_1$.} This state transition denotes that a target block is proposed after a previous timeout round. By the delay attack strategies in Sec.~\ref{subsec:delayattack}, the honest block will always be kept in the main chain. 
    Further, as the current system state is $S_1$, the delay for the target block to be committed is $\e{X_1}$.
    
    \item \emph{Case $b$: $S_{1} \xrightarrow{\beta} S_2$.} This state transition denotes that a target block together with its parent block can only form two consecutive blocks.
    According to the subsequent leaders (and blocks), the delay for this target block can be divided into three subcases:
    
    {\bf \emph{Subcase 1}}: The next leader is adversarial and proposes no block. This subcase happens with probability $\alpha$. By then, the system state is $S_0$, and the average delay for this target block to be committed is $\e{X_0} + 1$.
 
    {\bf \emph{Subcase 2}}: Some honest leader propose the next block, and then this block is overridden by a block of the subsequent adversarial leader. This subcase happens with probability $\beta \alpha$. By then, as the state is $S_1$, the average delay for this target block to be committed is $\e{X_1} + 2$.
    
    {\bf \emph{Subcase 3}}: Some honest leader propose the next two blocks. This subcase happens with probability $\beta^2$. As the state is $S_3$, the average delay for this target block to be committed is $\e{X_3} + 2$.

    \item \emph{Case $c$: $S_{2} \xrightarrow{\beta}S_3$ and $S_{3} \xrightarrow{\beta}S_4$.} This state transition denotes that a target block together with its predecessor blocks can form three consecutive blocks. According to the subsequent leaders (and blocks), the delay for this target block can be divided into two subcases:
    
     {\bf \emph{Subcase 1}}: Some honest leader propose a block, and then this new block is overridden by the next adversarial block. This subcase happens with probability $\beta \alpha$. As the system state is $S_1$, the average delay for this target block to be committed is $\e{X_1} + 2$.
    
    {\bf \emph{Subcase 2}}: Some honest leader propose two consecutive block. This subcase happens with probability $\beta^2$. As the system state is $S_3$, the average delay for this target block to be committed is $\e{X_3} + 2$.
\end{itemize}

\subsubsection{LibraBFT} \label{subsec:LibraLatency}
\begin{figure}[t]
\centering
\includegraphics[width=2.3in]{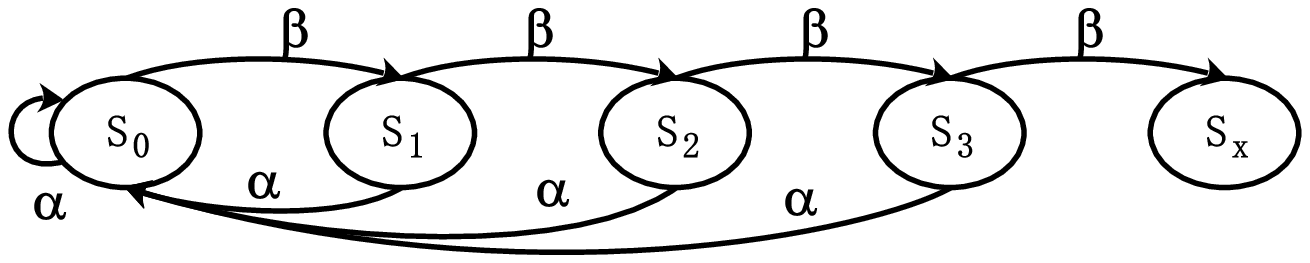}
\caption{\textbf{The state transition of the delay attack on pipelined HotStuff.}}
\vspace{-2mm}
\label{fig:latencyLibraBFT}
\end{figure}

By following a similar analysis of pipelined HotStuff, we first can show the state transitions for the system to enter the next state $S_x$ in Fig.~\ref{fig:latencyLibraBFT}. Furthermore, we can compute the expected rounds for a state $S_{i}$ ($i \in \{0, 1, 2, 3\}$) to first enter state $S_x$ and have the following equation.

\begin{equation} \label{eq:delayLibra}
\begin{split}
       \e{X_0} &= \alpha \e{X_{0}} + \beta \e{X_1} + 1 \\
       \e{X_1} &= \alpha \e{X_{0}} + \beta \e{X_2} + 1\\
       \e{X_2} &= \alpha \e{X_{0}} + \beta \e{X_3} + 1\\
       \e{X_3} &= \alpha \e{X_0} + 1\\
\end{split}
\end{equation}
By solving the equation~\eqref{eq:delayLibra}, we can get:
\begin{equation} \nonumber
\begin{split}
        \e{X_0} &= \frac{\beta^3 + \beta^2 + \beta + 1}{\beta^4}, \quad \e{X_1} = \frac{\beta^2 + \beta + 1}{\beta^4}, \\
 \e{X_2} &= \frac{\beta + 1}{\beta^4}, \quad \e{X_3} = \frac{1}{\beta^4} \\
\end{split}
\end{equation}

Next, we can track honest blocks kept in the main chain and their associated delay. According to its predecessor blocks structure, we can divide honest blocks into two cases as followings.
\begin{itemize}[leftmargin=*]
\setlength{\itemsep}{0pt}
\setlength{\topsep}{0pt}
\setlength{\partopsep}{0pt} 
    \item \emph{Case $a$: $S_0 \xrightarrow{\beta} S_1$ and $S_{1} \xrightarrow{\beta} S_2$.} By the delay attack strategies in Sec.~\ref{subsec:delayattack}, the honest block will always be kept in the main chain. Further, as the current system state is $S_1$, the delay for the honest block to be committed is $\e{X_1}$.
   
    \item \emph{Case $b$: $S_{2} \xrightarrow{\beta}S_3$ and $S_{3} \xrightarrow{\beta}S_4$.} This state transition denotes that an honest block together with its predecessor blocks can only form three consecutive blocks. According to the subsequent leaders (and blocks), the delay for this target block can be divided into two subcases:
    
     {\bf \emph{Subcase 1}}: Some honest leader propose a block, and then the subsequent adversarial leader hides the QC and proposes no block. This subcase happens with probability $\beta \alpha$. As the system state is $S_0$, the average delay for this target block to be committed is $\e{X_0} + 2$.
    
    {\bf \emph{Subcase 2}}: Some honest leader propose two consecutive block. This subcase happens with probability $\beta^2$. As the system state is $S_3$, the average delay for this target block to be committed is $\e{X_3} + 2$.
    
\end{itemize}

\begin{figure}[t]
\centering
\includegraphics[width=2.3in]{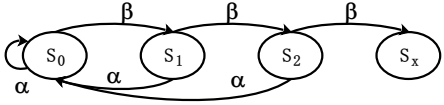}
\caption{\textbf{The state transition of the delay attack on pipelined HotStuff.}}
\vspace{-2mm}
\label{fig:latencyBroadcastingQC}
\end{figure}

\subsubsection{Broadcasting QC} \label{subsec:QCLatency}
By following the previous analysis, we first can show the state transitions for the system to enter the next state $S_x$ in Fig.~\ref{fig:latencyBroadcastingQC}. Furthermore, we can compute the expected rounds for a state $S_{i}$ ($i \in \{0, 1, 2, 3\}$) to first enter state $S_x$, and have the following equation.
\begin{equation} \label{eq:delay}
\begin{split}
       \e{X_0} &= \alpha \e{X_{0}} + \beta \e{X_1} + 1 \\
       \e{X_1} &= \alpha \e{X_{0}} + \beta \e{X_2} + 1\\
       \e{X_2} &= \alpha \e{X_{0}} + 1\\
\end{split}
\end{equation}
By solving the equation~\eqref{eq:delay}, we can get:
\begin{equation} \nonumber
\begin{split}
        \e{X_0} = \frac{\beta^2 + \beta + 1}{\beta^3},\quad \e{X_1} = \frac{\beta + 1}{\beta^3},\quad \e{X_2} = \frac{1}{\beta^3}.
\end{split}
\end{equation}
Next, we can track honest blocks kept in the main chain and their associated delay. Note that all honest blocks will be kept in the main chain by the delay attack strategies in pipelined HotStuff with broadcasting QCs. Furthermore, as the system state is $S_1$, the delay for the honest block to be committed is $\e{X_1}$.
\end{document}